\documentclass[11pt, draftcls, onecolumn]{IEEEtran}
\usepackage{subfigure}
\usepackage{setspace}
\usepackage{amsmath}
\usepackage{amssymb}
\usepackage{amsfonts}
\usepackage{amscd}
\usepackage{mathrsfs}
\usepackage[final]{graphicx}
\usepackage{graphicx}
\usepackage{psfrag}
\usepackage{epsfig}
\usepackage{color}
\usepackage{url}
\usepackage{textcomp}
\usepackage{multirow}
\input{epsf.sty}
\newtheorem{theorem}{Theorem}
\newtheorem{lemma}{Lemma}
\newtheorem{definition}{Definition}

\begin{document}

\title{Generalized Silver Codes}
\author{
\authorblockN{K. Pavan Srinath and B. Sundar Rajan,\\}
\authorblockA{Dept of ECE, Indian Institute of science, \\
Bangalore 560012, India\\
Email:\{pavan,bsrajan\}@ece.iisc.ernet.in\\
}
}
\maketitle
\vspace{-15mm}
\begin{abstract}
For an $n_t$ transmit, $n_r$ receive antenna system ($n_t \times n_r$ system), a {\it{full-rate}} space time block code (STBC) transmits $n_{min} = min(n_t,n_r)$ complex symbols per channel use. The well known Golden code is an example of a full-rate, full-diversity STBC for 2 transmit antennas. Its ML-decoding complexity is of the order of $M^{2.5}$ for square $M$-QAM. The Silver code for 2 transmit antennas has all the desirable properties of the Golden code except its coding gain, but offers lower ML-decoding complexity of the order of $M^2$. Importantly, the slight loss in coding gain is negligible compared to the advantage it offers in terms of lowering the ML-decoding complexity. For higher number of transmit antennas, the best known codes are the Perfect codes, which are full-rate, full-diversity, information lossless codes (for $n_r \geq n_t$) but have a high ML-decoding complexity of the order of $M^{n_tn_{min}}$ (for $n_r < n_t$, the punctured Perfect codes are considered). In this paper\footnote[1]{Part of the content of this manuscript has been presented at IEEE ISIT 2010 and another part at IEEE Globecom, 2010.}, a scheme to obtain full-rate STBCs for $2^a$ transmit antennas and any $n_r$ with reduced ML-decoding complexity of the order of $M^{n_t(n_{min}-\frac{3}{4})-0.5}$, is presented. The codes constructed are also information lossless for $n_r \geq n_t$, like the Perfect codes and allow higher mutual information than the comparable punctured Perfect codes for $n_r < n_t$. These codes are referred to as the {\it generalized Silver codes}, since they enjoy the same desirable properties as the comparable Perfect codes (except possibly the coding gain) with lower ML-decoding complexity, analogous to the Silver-Golden codes for 2 transmit antennas. Simulation results of the symbol error rates for 4 and 8 transmit antennas show that the generalized Silver codes match the punctured Perfect codes in error performance while offering lower ML-decoding complexity.
\end{abstract}

\hrule 

\section{Introduction and Background}
Complex orthogonal designs (CODs) \cite{TJC}, \cite{TiH}, although provide linear Maximum Likelihood (ML)-decoding, do not offer a high rate of transmission. A full-rate code for an $n_t \times n_r$ MIMO system transmits $min(n_t,n_r)$ independent complex symbols per channel use. Among the CODs, only the Alamouti code for 2 transmit antennas is full-rate for a $2 \times 1$ MIMO system. A full-rate STBC can efficiently utilize all the degrees of freedom the channel provides. In general, an increase in the rate tends to result in an increase in the ML-decoding complexity. The Golden code \cite{BRV} for 2 transmit antennas is an example of a full-rate STBC for any number of receive antennas. Until recently, the ML-decoding complexity of the Golden code was reported to be of the order of $M^4$, where $M$ is the size of the signal constellation. However, it was shown in \cite{SrR_arxiv}, \cite{john_barry1} that the Golden code has a decoding complexity of the order of $M^{2.5}$ for square $M$-QAM. Current research focuses on obtaining high rate codes with {\it reduced ML-decoding complexity} (refer to Sec. \ref{sec2} for a formal definition). For 2 transmit antennas, the Silver code, named so in \cite{Hollanti_silver},  was first mentioned in \cite{tirk_combined} and independently presented in \cite{PGA} along with a study of its low ML-decoding complexity property. It is a full-rate code with full-diversity and an ML-decoding complexity of the order of $M^2$ for square $M$-QAM. Its algebraic properties have been studied in \cite{Hollanti_silver} and \cite{Ray_silver} and a fixed point fast decoding scheme has been given in \cite{fixed_point_silver}. For 4 transmit antennas, Biglieri et. al. proposed a rate-2 STBC which has an ML-decoding complexity of the order of $M^{4.5}$ for square $M$-QAM without full-diversity \cite{BHV}. It was, however, shown that there was no significant reduction in error performance at low to medium SNR when compared with the previously best known code - the DjABBA code \cite{HTW}. This code was obtained by multiplexing Quasi-orthogonal designs (QOD) for 4 transmit antennas \cite{JH}. In \cite{SrR_arxiv}, a new full-rate STBC for $4\times 2$ system with an ML-decoding complexity of $M^{4.5}$ was proposed and was conjectured to have the non-vanishing determinant (NVD) property. This code was obtained by multiplexing the coordinate interleaved orthogonal designs (CIODs) for 4 transmit antennas \cite{ZS}. These results show that codes obtained by multiplexing low complexity STBCs can result in high rate STBCs with reduced ML-decoding complexity and by choosing a suitable constellation, there won't be any significant degradation in the error performance when compared with the best existing STBCs. Such an approach has also been adopted in \cite{Robert} to obtain high rate codes\footnote{Fast decodable STBCs have been constructed in \cite{fgd1}-\cite{fgd3}, but these codes are not full-rate in general, and make use of near ML-decoding algorithms.} from multiplexed orthogonal designs. More recently, full-rate STBCs with an ML-decoding complexity of the order of $M^{5.5}$ and a provable NVD property for the $4\times2$ system have been proposed in \cite{Oggier_fast} and \cite{Oggier_spcom}.

In general, it is not known how one can design full-rate STBCs for an arbitrary number of transmit and receive antennas with reduced ML-decoding complexity. It is well known that the maximum mutual information achievable with an STBC is at best equal to the ergodic capacity of the MIMO channel, in which case the STBC is said to be {\it information lossless} (see Section \ref{sec2} for a formal definition). It is known how to design information lossless codes \cite{damen_info} for the case where $n_r \geq n_t$. However, when $n_r < n_t$ the only known code in literature which is information lossless is the Alamouti code, which is information lossless for the $ 2 \times 1 $ system alone. It has been shown in \cite{tirk_combined}, \cite{tirk_expansion} and \cite{HTW} that when $n_r < n_t$, {\it self-interference} of the STBC (a formal definition of self interference is given in Section \ref{sec2}) has to be minimized for maximizing the mutual information achieved with the STBC. Not much research\footnote{The full-rate STBCs in \cite{hollanti1}, designed for $n_r < n_t$, are not linear dispersion codes. They are based on maximal orders and use spherical shaping due to which the encoding and decoding complexity is extremely high. The STBCs in \cite{hollanti2}, also designed for $n_r < n_t$, use the concept of restricting the number of active transmit antennas to be no larger than the number of receive antennas, and so, the mutual information analysis for these codes is very difficult. These STBCs are diversity-multiplexing gain tradeoff (DMT) optimal but are associated with a very high ML-decoding complexity.} has been done on designing codes that allow a high mutual information when $n_r < n_t$. In this paper, for $n_t = 2^a$, we systematically design full-rate STBCs which have the least possible self-interference and the lowest ML-decoding complexity among known full-rate STBCs for $n_r < n_t$ and consequently allow higher mutual information than the best existing codes (the Perfect codes with puncturing \cite{ORBV}, \cite{new_per}), while for $n_r \geq n_t$, the proposed STBCs are information lossless like the comparable Perfect codes. We call these codes the {\it generalized Silver codes} since, analogous to the silver code and the Golden code for 2 transmit antennas, the proposed codes have every desirable property that the Perfect codes have, except the coding gain, but importantly, have lower ML-decoding complexity. The contributions of the paper are:

\begin{enumerate}
\item We give a scheme to obtain rate-1, 4-group decodable codes (refer Section \ref{sec2} for a formal definition of multi-group decodable codes) for $n_t = 2^a$ through algebraic methods. The speciality of the obtained design is that it is amenable for extension to higher number of receive antennas, resulting in full-rate codes with reduced ML-decoding complexity for any number of receive antennas, unlike the previous constructions \cite{4gp1}-\cite{sanjay} of rate-1, 4-group decodable codes.
\item Using the rate-1, 4-group decodable codes thus constructed, we propose a scheme to obtain the generalized Silver codes, which are full-rate codes with reduced ML-decoding complexity for $2^a$ transmit antennas and any number of receive antennas. These codes also have the least self-interference among known comparable STBCs and allow higher mutual information with lower ML-decoding complexity than the comparable punctured Perfect codes for the case $n_r < n_t$, while being information lossless for $n_r \geq n_t$. In terms of error performance, by choosing the signal constellation carefully, the proposed codes have more or less the same performance as the corresponding punctured Perfect codes. This is shown through simulation results for 4 and 8 transmit antenna systems.
\end{enumerate}

The paper is organized as follows. In Section \ref{sec2}, we present the system model and the relevant definitions. The criteria for maximizing the mutual information with space time modulation are presented in Section \ref{sec3} and our method to construct rate-1, 4-group decodable codes is proposed in Section \ref{sec4}. The scheme to extend these codes to obtain the generalized Silver codes for higher number of receive antennas is presented in Section \ref{sec5}. Simulation results are discussed in Section \ref{sec6} and the concluding remarks are made in Section \ref{sec7}.

{\it Notations:} Throughout, bold, lowercase letters are used to denote vectors and bold, uppercase letters are used to denote matrices. Let $\textbf{X}$ be a complex matrix. Then, $\textbf{X}^{H}$ and $\textbf{X}^{T}$ denote the Hermitian and the transpose of $\textbf{X}$, respectively and unless used to denote indices or subscripts, $j$ represents $\sqrt{-1}$. The $(i,j)^{th}$ entry of $\textbf{X}$ is denoted by $\textbf{X}(i,j)$ while $tr(\textbf{X})$ and $det(\textbf{X})$ denote the trace and determinant of $\textbf{X}$, respectively. The set of all real and complex numbers are denoted by $\mathbb{R}$ and $\mathbb{C}$, respectively. The real and the imaginary part of a complex number $x$ are denoted by $x_I$ and $x_Q$, respectively. $\Vert \textbf{X} \Vert$ denotes the Frobenius norm of $\textbf{X}$, $\Vert \textbf{x} \Vert$ denotes the vector norm of a vector $\textbf{x}$, and $\textbf{I}_T$ and $\textbf{O}_T$ denote the $T\times T$ identity matrix and the null matrix, respectively. The Kronecker product is denoted by $\otimes$ and $vec(\textbf{X})$ denotes the concatenation of the columns of $\textbf{X}$ one below the other. For a complex random variable $X$, $\mathbb{E}[X]$ denotes the mean of $X$ and $\mathbb{E}_{X}\left(f(X)\right)$ denotes the mean of $f(X)$, a function of the random variable $X$. The inner product of two vectors $\textbf{x}$ and $\textbf{y}$ is denoted by $\langle \textbf{x},\textbf{y} \rangle$. For a set $\mathcal{S}$, $a\mathcal{S} \triangleq \{as | s \in \mathcal{S} \}$. Let $\mathcal{P}$ and $\mathcal{Q}$ be two sets such that $\mathcal{P} \supset \mathcal{Q}$. Then $\mathcal{P} \setminus \mathcal{Q}$ denotes the set of elements of $\mathcal{P}$ excluding the elements of $\mathcal{Q}$. For a complex variable $x$, the $\check{(.)}$ operator acting on $x$ is defined as
\begin{equation*}
\check{x} \triangleq \left[ \begin{array}{rr}
                             x_I & -x_Q \\
                             x_Q & x_I \\
                            \end{array}\right].
\end{equation*}
The $\check{(.)}$ can similarly be applied to any matrix $\textbf{X} \in \mathbb{C}^{n \times m}$ by replacing each entry $x_{ij}$ with $\check{x}_{ij}$, $i=1,2,\cdots, n, j = 1,2,\cdots,m$, resulting in a matrix denoted by $\check{\textbf{X}} \in \mathbb{R}^{2n \times 2m}$. Given a complex vector $\textbf{x} = [ x_1, x_2, \cdots, x_n ]^T$, $\tilde{\textbf{x}}$ is defined as $\tilde{\textbf{x}} \triangleq [ x_{1I},x_{1Q}, \cdots, x_{nI}, x_{nQ} ]^T$. It follows that for $\textbf{A} \in \mathbb{C}^{m \times n}$, $ \textbf{B} \in \mathbb{C}^{n  \times p}$ and $\textbf{C} = \textbf{AB}$, the equalities $ \check{\textbf{C}} =  \check{\textbf{A}}\check{\textbf{B}}$ and $\widetilde{vec(\textbf{C})}  =  (\textbf{I}_p \otimes \check{\textbf{A}})\widetilde{vec(\textbf{B})}$ hold.

\section{System Model}
\label{sec2}
We consider the Rayleigh block fading MIMO channel with full channel state information (CSI) at the receiver but not at the transmitter. For $n_t \times n_r$ MIMO transmission, we have
\begin{equation}\label{Y}
\textbf{Y} = \sqrt{\frac{SNR}{n_t}}\textbf{HS + N},
\end{equation}

\noindent where $\textbf{S} \in \mathbb{C}^{n_t \times T}$ is the codeword matrix whose average energy is given by $\mathbb{E}(\Vert \textbf{S} \Vert^2) = n_tT$, $\textbf{N} \in \mathbb{C}^{n_r \times T}$ is a complex white Gaussian noise matrix with i.i.d. entries $\sim
\mathcal{N}_{\mathbb{C}}\left(0,1\right)$ (complex normal distribution with zero mean and unit variance), $\textbf{H} \in \mathbb{C}^{n_r\times n_t}$ is the channel matrix with the entries assumed to be i.i.d. circularly symmetric Gaussian random variables $\sim \mathcal{N}_\mathbb{C}\left(0,1\right)$, $\textbf{Y} \in \mathbb{C}^{n_r \times T}$ is the received matrix and $SNR$ is the signal-to-noise ratio at each receive antenna.

\begin{definition}\label{def1}({\it Code rate}) Code rate is the average number of independent information symbols transmitted per channel use. If there are $k$ independent complex information symbols (or $2k$ real information symbols) in the codeword which are transmitted over $T$ channel uses, then, the code rate is $k/T$ complex symbols per channel use ($2k/T$ real symbols per channel use).
\end{definition}

\begin{definition}\label{def2}({\it Full-rate STBCs}) For an $n_t \times n_r$ MIMO system, if the code rate is $min\left(n_t,n_r\right)$ complex symbols per channel use, then the STBC is said to be \emph{{full-rate}}.
\end{definition}

 Assuming ML-decoding, the metric that is to be minimized over all possible values of codewords $\textbf{S}$ is given by
 \begin{equation*}
 \textbf{M}\left(\textbf{S}\right) = \left \Vert \textbf{Y} - \sqrt{\frac{SNR}{n_t}}{}\textbf{HS} \right \Vert^2.
 \end{equation*}

\begin{definition}\label{def3}({\it ML-Decoding complexity})
The ML decoding complexity is measured in terms of the maximum number of symbols that need to be jointly decoded in minimizing the ML decoding metric.
\end{definition}
For example, if the codeword transmits $k$ independent symbols of which a maximum of $p$ symbols need to be jointly decoded, the ML-decoding complexity is of the order of $M^{p}$, where $M$ is the size of the signal constellation. If the code has an ML-decoding complexity of order less than $M^k$, the code is said to have \emph{{reduced ML-decoding}} complexity.

\begin{definition}\label{def4}({\it Generator matrix}) For any STBC that encodes $2k$ real symbols (or $k$ complex information symbols), the {\emph{generator}} matrix $\textbf{G} \in \mathbb{R}^{2Tn_t \times 2k}$ is defined by \cite{BHV}
\begin{equation*}
\widetilde{vec\left(\textbf{S}\right)} = \textbf{G} \textbf{s},
\end{equation*}
\noindent where $\textbf{S}$ is the codeword matrix, $\textbf{s} \triangleq \left[ s_1, s_2,\cdots,s_{2k} \right]^T$ is the real information symbol vector.
\end{definition}

A codeword matrix of an STBC can be expressed in terms of \emph{weight matrices} (linear dispersion matrices) \cite{HaH} as
\begin{equation*}
\textbf{S} = \sum_{i=1}^{2k}s_{i}\textbf{A}_{i} .
\end{equation*}
Here, $\textbf{A}_i,i=1,2,\cdots,2k$, are the complex weight matrices of the STBC and should form a {\it linearly independent} set over $\mathbb{R}$. It follows that
\begin{equation*}
\textbf{G} = \left[\widetilde{vec(\textbf{A}_1)}\ \widetilde{vec(\textbf{A}_2)}\ \cdots \ \widetilde{vec(\textbf{A}_{2k})}\right].
\end{equation*}

Due to the constraint that $\mathbb{E}\left(\Vert \textbf{S} \Vert^2 \right) = n_tT$, we have, $\sum_{i=1}^{2k}\mathbb{E}(s_i)^2tr\left(\textbf{A}_i\textbf{A}_i^H\right) = n_tT$. Choosing $\mathbb{E}(s_i)^2 = 1/2$ for all $i=1,\cdots,2k$, we have 
\begin{equation}\label{constraint}
\sum_{i=1}^{2k}tr\left(\textbf{A}_i\textbf{A}_i^H\right) = 2n_tT.
\end{equation}

\begin{definition}\label{def5}({\it Multi-group decodable STBCs}) An STBC is said to be $g$-group decodable \cite{sanjay} if its weight matrices can be separated into $g$ groups $\mathcal{G}_1$, $\mathcal{G}_2$, $\cdots$, $\mathcal{G}_g$ such that
\begin{equation*}
\textbf{A}_i\textbf{A}_j^H + \textbf{A}_j\textbf{A}_i^H = \textbf{O}_{n_t}, ~~~~ \textbf{A}_i \in \mathcal{G}_l, ~~~ \textbf{A}_j \in \mathcal{G}_p, ~~~ l,p \in \{1,2,\cdots,g\}, ~~ l \neq p.
\end{equation*}
\end{definition}

\begin{definition}\label{def8}({\it Self-interference}) For an STBC given by $\textbf{S} = \sum_{i=1}^{2k}s_{i}\textbf{A}_{i}$, the self-interference matrix \cite{HTW} is defined as 
\begin{equation*}
 \textbf{S}^{int} = \sum_{i=1}^{2k-1} \sum_{j>i}^{2k}s_is_j\left(\textbf{A}_i\textbf{A}_j^H + \textbf{A}_j\textbf{A}_i^H \right).
\end{equation*}
\end{definition}

\begin{definition}\label{def7}({\it Punctured Codes}) Punctured STBCs are the codes with some of the symbols being zeros, in order to meet the full-rate criterion.
\end{definition}

For example, a codeword of the Perfect code for 4 transmit antennas \cite{ORBV} transmits sixteen complex symbols in four channel uses and has a rate of 4 complex symbols per channel use. If this code were to be used for a two receive antenna system which can only support a rate of two independent complex symbols per channel use, then, eight symbols of the Perfect code can be made zeros, so that the codeword transmits eight complex symbols in four channel uses. These eight symbols correspond to the two layers \cite{ORBV} of the Perfect code.

Equation \eqref{Y} can be rewritten as
\begin{equation}\label{eq_model}
 \widetilde{vec(\textbf{Y})} = \sqrt{\frac{SNR}{n_t}}\textbf{H}_{eq}\textbf{s} + \widetilde{vec(\textbf{N})},
\end{equation}
\noindent where $\textbf{H}_{eq} \in \mathbb{R}^{2n_rT\times 2n_{min}T}$, called the { \it equivalent channel matrix}. is given by $\textbf{H}_{eq} = \left(\textbf{I}_T \otimes \check{\textbf{H}}\right)\textbf{G}$, with $\textbf{G} \in \mathbb{R}^{2n_tT\times 2n_{min}T}$ being the generator matrix as in Definition \ref{def4}.

\begin{definition}\label{def6}({\it Ergodic capacity}) The ergodic capacity of an $n_t \times n_r$ MIMO channel is \cite{tel}
\begin{equation*}
 \mathcal{C}_{n_t \times n_r} = \mathbb{E}_{\textbf{H}}\left(log~det\left(\textbf{I}_{n_r} + \frac{SNR}{n_t}\textbf{H}\textbf{H}^H\right)\right).
\end{equation*}
\end{definition}

With the use of an STBC, the maximum mutual information achievable is \cite{JJK}
\begin{equation*}
 \mathcal{I}_{STBC} = \frac{1}{2T}\mathbb{E}_{\textbf{H}}\left(log~det\left(\textbf{I}_{2n_rT} + \frac{SNR}{n_t}\textbf{H}_{eq}\textbf{H}_{eq}^T\right)\right).
\end{equation*}

It is known that $\mathcal{C}_{n_t \times n_r} \geq \mathcal{I}_{STBC}$. If $\mathcal{C}_{n_t \times n_r} = \mathcal{I}_{STBC}$, the STBC is said to be {\it information lossless}. If the generator matrix $\textbf{G}$ is orthogonal (from Definition \ref{def4}, this case arises only if $n_r \geq n_t$ and the STBC is full-rate, i.e, $k = n_tT$), the STBC is information lossless.

\section{Relationship between weight matrices and the maximum mutual information}\label{sec3}

Capacity can be achieved with the use of continuous inputs with Gaussian distribution. If one were able to use continuous Gaussian distributed inputs in practice, using the V-blast scheme would suffice, since diversity is irrelevant. But in practice, one has to use finite discrete inputs, and diversity becomes an important aspect, necessitating the use of full-diversity STBCs. Even though we considered the limited block length scenario for space-time coding as a standalone scheme, in practice, one would also have an outer code and coding would be done over large block lengths to go close to capacity. In such a scenario, the maximum mutual information that an STBC allows becomes an important parameter for the design of STBCs. It is preferable to use STBCs which allow mutual information as close to the channel capacity as possible. It has been shown that if the generator matrix is orthogonal, the maximum mutual information achievable with the STBC is the same as the ergodic capacity of the MIMO channel \cite{damen_info}, \cite{JJK}. For the generator matrix to be orthogonal, a prerequisite is that the number of receive antennas should be at least equal to the number of transmit antennas. When $n_r < n_t$, only the Alamouti code has been known to be information lossless for the $2 \times 1$ MIMO channel. In \cite{tirk_expansion}, by using the well known matrix identities $det ~\textbf{M} = e^{tr (log \textbf{M})}$ and $log(\textbf{I}+\textbf{X}) = \sum_{n=1}^{\infty}\frac{(-1)^{n-1}}{n}\textbf{X}^n$, an expansion of the ergodic MIMO capacity in SNR was obtained as 
\begin{eqnarray*}
 \mathcal{C}_{n_t\times n_r}&  = &\sum_{n=1}^{\infty} C_n SNR^n, 
\end{eqnarray*}
with $C_n = \frac{-1}{n}\left(\frac{-1}{n_t}\right)^n \mathbb{E}_{\textbf{H}} \left(tr \left[ \left (\textbf{HH}^H \right)^n\right]\right)$. The first two coefficients can easily be checked to be $C_1 = n_r$ and $C_2 = -n_r(n_r+n_t)/n_t$. On a similar note, $\mathcal{I}_{STBC}$ can also be expanded in SNR as $\mathcal{I}_{STBC}  = \sum_{n=1}^{\infty} I_n SNR^n$, where 

\begin{equation}\label{defn_I}
I_n = \frac{-1}{2Tn}\left(\frac{-1}{n_t}\right)^n \mathbb{E}_{\textbf{H}} \left(tr \left[ \left (\textbf{H}_{eq}\textbf{H}_{eq}^T \right)^n\right]\right) = \frac{-1}{2Tn}\left(\frac{-1}{n_t}\right)^n \mathbb{E}_{\textbf{H}} \left(tr \left[ \left(\textbf{H}_{eq}^T\textbf{H}_{eq}\right)^n\right]\right).
\end{equation}

Let $\bar{\textbf{H}} \triangleq \textbf{H}_{eq}^T\textbf{H}_{eq}$. It is straightforward to check that $\bar{\textbf{H}}(i,j) = \frac{1}{2}tr \left(\textbf{S}_{ij}\textbf{H}^H\textbf{H} \right)$, where $\textbf{S}_{ij} \triangleq \textbf{A}_i\textbf{A}_j^H + \textbf{A}_j\textbf{A}_i^H$. Hence, 
\begin{eqnarray*}
 I_1 & = &\frac{1}{2Tn_t}\mathbb{E}_{\textbf{H}}(tr\left[\textbf{H}_{eq}^T\textbf{H}_{eq}\right]) = \frac{1}{4Tn_t}\sum_{i=1}^{2Tn_r}\mathbb{E}_{\textbf{H}} \left( tr \left(\textbf{S}_{ii}\textbf{H}^H\textbf{H} \right)\right)\\
& = & \frac{1}{2Tn_t}\sum_{i=1}^{2Tn_r} tr \left(\textbf{A}_i\textbf{A}_i^H\mathbb{E}_{\textbf{H}}\left(\textbf{H}^H\textbf{H} \right)\right) = n_r,
\end{eqnarray*}
where $\mathbb{E}\left(\textbf{H}^H\textbf{H} \right) = n_r \textbf{I}_{n_t}$ and \eqref{constraint} is used in obtaining $I_1$. So, using all the available power helps one to achieve the first order capacity. The second coefficient $I_2$ has been calculated in \cite{tirk_combined} to be
\begin{equation}\label{coefficient_2}
 I_2 = \frac{-n_r}{16Tn_t^2} \sum_{i=1}^{2Tn_r}\sum_{j=i}^{2Tn_r}\left(tr(\textbf{S}_{ij}^2) + n_r(tr \textbf{S}_{ij})^2 \right).
\end{equation}

In \cite{tirk_combined}, it was argued that typical discrete input schemes fail to achieve capacity at the third order in the expansion of the mutual information and hence, $I_2$ should be maximized. From \eqref{coefficient_2}, it is clear that to maximize $I_2$, the following criteria should be satisfied.
\begin{enumerate}
\item {\it Hurwitz-Radon Orthogonality}: as many of $\textbf{S}_{ij}$ should be equal to $\textbf{O}_{n_t}$ as possible, for $1 \leq i <j \leq 2Tn_r$.
 \item{ \it Tracelessness}: $\textbf{S}_{ij}$ should be traceless, for all $1 \leq i <j \leq 2Tn_r$.
\end{enumerate}

In fact, the first criterion, which is equivalent to minimizing the self-interference, is already clear from \eqref{defn_I}, where it can be observed that a larger number of zero entries of $\textbf{H}_{eq}^T\textbf{H}_{eq}$ contributes to a lower value of the trace of $\left(\textbf{H}_{eq}^T\textbf{H}_{eq}\right)^2$. Hence, to design a good STBC with a high mutual information when $n_r < n_t$, one should have as many as possible weight matrix pairs satisfying Hurwitz-Radon (HR) orthogonality. We would, of course, like all the weight matrices to satisfy HR-orthogonality, but there is a limit to this number \cite{TJC} which, except for the Alamouti code, is much lesser than $2Tn_r$, the number of weight matrices of a full-rate STBC when $n_r < n_t$. It can easily be checked that for the Alamouti code, $I_2 = C_2$. It is known that for a rate-1 code for $n_t > 2$, one cannot have all the full-ranked weight matrices mutually satisfying HR-orthogonality. For such STBCs, the minimum self-interference is achieved if the STBCs are $g$-group decodable, with $g$ as large as possible.  At present, the best known rate-1 low complexity multi-group decodable codes are the 4-group decodable codes for any number of transmit antennas \cite{4gp1}, \cite{4gp2}, \cite{sanjay}. These codes are not full-rate for $n_r > 1$. If one were to require a full-rate code, the codes in literature  \cite{4gp1}, \cite{4gp2}, \cite{sanjay} are not suitable for extension to higher number of receive antennas, since their design is obtained by iterative methods. In the next section, we propose a new design methodology to obtain the weight matrices of a rate-1, 4-group decodable code by algebraic methods for $2^a$ transmit antennas. These codes can be extended to higher number of receive antennas to obtain full-rate STBCs with lower ML-decoding complexity and lower self-interference than the existing designs.

\section{Construction of Rate-1, 4-group decodable codes}\label{sec4}
We make use of the following theorem, presented in \cite{4gp2}, to construct rate-1, 4-group decodable codes for $n = 2^a$ transmit antennas.
\begin{theorem}
\label{thm1}
\cite{4gp2} An $n \times n$ linear dispersion code transmitting k real symbols is $g$-group decodable if the weight matrices satisfy the following conditions:
\begin{enumerate}
\item $\textbf{A}_i^2   =  \textbf{I}_n, ~~ i \in \{1,2,\cdots,\frac{k}{g}\}$.
\item $\textbf{A}_j^2   =  -\textbf{I}_n, ~~ j \in \{\frac{mk}{g} + 1, m = 1,2,\cdots,g-1\}$.
\item $\textbf{A}_i\textbf{A}_j = \textbf{A}_j\textbf{A}_i, ~~ i,j \in \{1,2,\cdots,\frac{k}{g} \}$.
\item $\textbf{A}_i\textbf{A}_j =  \textbf{A}_j\textbf{A}_i, ~~ i \in \{1,2,\cdots,\frac{k}{g} \}, ~~ j \in \{ \frac{mk}{g} + 1, m = 1,2,\cdots,g-1\}$.
\item $\textbf{A}_i\textbf{A}_j  =  -\textbf{A}_j\textbf{A}_i, ~~ i, j \in \{  \frac{mk}{g} + 1, m = 1,2,\cdots,g-1\}, ~ i \neq j$.
\item $\textbf{A}_{\frac{mk}{g} + i} = \textbf{A}_i\textbf{A}_{\frac{mk}{g}+1}, ~ m \in \{ 1,2,\cdots,g-1\}$, ~ $i \in \{1,2,\cdots,\frac{k}{g} \} $.
\end{enumerate}

Table \ref{table1} illustrates the weight matrices of a $g$-group decodable code which satisfy the above conditions. The weight matrices in each column belong to the same group.
\begin{table*}
\begin{center}
\begin{tabular}{|l|l|l|l|}
\hline
$\textbf{A}_1 = \textbf{I}_n$        &  $\textbf{A}_{\frac{k}{g}+1}$ & $\ldots$ &  $\textbf{A}_{\frac{(g-1)k}{g}+1}$ \\ \hline
$\textbf{A}_2$              &  $\textbf{A}_{\frac{k}{g}+2} = \textbf{A}_2\textbf{A}_{\frac{k}{g}+1}$ & $\ldots$ & $\textbf{A}_{\frac{(g-1)k}{g}+2} = \textbf{A}_2\textbf{A}_{\frac{(g-1)k}{g}+1}$ \\ \hline
$ \vdots$    & $\vdots$ & $\ldots$ & $\vdots$  \\ \hline
$\textbf{A}_{\frac{k}{g}}$  & $\textbf{A}_{\frac{2k}{g}} = \textbf{A}_{\frac{k}{g}}\textbf{A}_{\frac{k}{g}+1}$ & $\ldots$ & $\textbf{A}_{k} = \textbf{A}_{\frac{k}{g}}\textbf{A}_{\frac{(g-1)k}{g}+1}$   \\ \hline
\end{tabular}
\caption{Weight matrices of a $g$-group decodable code}
\label{table1}
\end{center}
\hrule
\end{table*}

\end{theorem}

In order to obtain a rate-1, 4-group decodable STBC for $2^a$ transmit antennas, it is sufficient if we have $2^{a+1}$ matrices satisfying the conditions in Theorem \ref{thm1}. To obtain these\footnote{These STBCs can be obtained elegantly using the theory of Clifford Algebra but to make the paper accessible to a wider group of readers, we have preferred to make use of simple concepts from matrix theory without reference to Clifford Algebra.}, we make use of the following lemmas.

\begin{lemma}\label{lemma_1}\cite{anti_matric}
If $n=2^a$ and invertible complex matrices of size $n \times n$, denoted by $\textbf{F}_i, i=1,2,\cdots,2a$, anticommute pairwise, then the set of products $\textbf{F}_{i_1}\textbf{F}_{i_2}\cdots \textbf{F}_{i_s}$ with $1 \leq i_1 < \cdots < i_s \leq 2a$ along with $\textbf{I}_n$ forms a basis for the $2^{2a}$ dimensional space of all $n \times n$ matrices over $\mathbb{C}$.
\end{lemma}
\begin{proof}
The proof is provided for the sake of completeness. Assume that in the set of products $\textbf{F}_{i_1}\textbf{F}_{i_2}\cdots \textbf{F}_{i_s}$, $1 \leq i_1 < \cdots < i_s \leq 2a$, along with $\textbf{I}_{n}$, at most $k$ elements are linearly independent over $\mathbb{C}$, for some $k < 2^{2a}$. So,
\begin{equation}\label{eq_basis}
 \sum_{i=1}^{k+1} \alpha_i \textbf{F}_1^{\lambda_{i_1}}\textbf{F}_2^{\lambda_{i_2}}\cdots\textbf{F}_{2a}^{\lambda_{i_{2a}}} = \textbf{O}_n, ~~ \alpha_i \neq 0, ~~ \lambda_{i_j} \in \{0,1\}, j = 1,2,\cdots,2a.
\end{equation}
Noting that $\textbf{F}_2\cdots\textbf{F}_{2a}$ anticommutes with $\textbf{F}_1$ but commutes with each of $\textbf{F}_2$, $\cdots$, $\textbf{F}_{2a}$, premultiplying each term of \eqref{eq_basis} by $\textbf{F}_2\cdots\textbf{F}_{2a}$ results in a new equation with the coefficients $\alpha_i$ negated for those terms in \eqref{eq_basis} containing $\textbf{F}_1$. Adding this new equation to \eqref{eq_basis} yields another equation containing fewer summands than \eqref{eq_basis}, leading to a contradiction. So, $k=2^{2a}$, which proves the theorem. 
\end{proof}

\begin{lemma}\label{lemma_2}
If all the mutually anticommuting $n \times n$ matrices $\textbf{F}_i, i=1,2,\cdots,2a$, are unitary and anti-Hermitian, so that they square to $-\textbf{I}_n$, then the product $\textbf{F}_{i_1}\textbf{F}_{i_2}\cdots \textbf{F}_{i_s}$ with $1 \leq i_1 < \cdots < i_s \leq 2a$ squares to $(-1)^{\frac{s(s+1)}{2}}\textbf{I}_n$.
\end{lemma}
\begin{proof}
We have
\begin{eqnarray*}
(\textbf{F}_{i_1}\textbf{F}_{i_2} \cdots \textbf{F}_{i_s})(\textbf{F}_{i_1}\textbf{F}_{i_2} \cdots \textbf{F}_{i_s}) & = & (-1)^{s-1}(\textbf{F}_{i_1}^2 \textbf{F}_{i_2} \cdots \textbf{F}_{i_s})(\textbf{F}_{i_2}\textbf{F}_{i_3} \cdots \textbf{F}_{i_s})\\
& = & (-1)^{s-1}(-1)^{s-2}(\textbf{F}_{i_1}^2 \textbf{F}_{i_2}^2 \cdots \textbf{F}_{i_s})(\textbf{F}_{i_3}\textbf{F}_{i_4} \cdots \textbf{F}_{i_s}) \\
& = & (-1)^{[(s-1)+(s-2)+\cdots 1]}(\textbf{F}_{i_1}^2 \textbf{F}_{i_2}^2 \cdots \textbf{F}_{i_s}^2)\\
& = & (-1)^{\frac{s(s-1)}{2}}(-1)^{s}\textbf{I}_n  =  (-1)^{\frac{s(s+1)}{2}}\textbf{I}_n,
\end{eqnarray*}
which proves the lemma.
\end{proof}

\begin{lemma}\label{lemma_3}
Let $\textbf{F}_i, i = 1,2,\cdots,2a$ be anticommuting, anti-Hermitian, unitary matrices. Let $\Omega_1 = \{ \textbf{F}_{i_1},\textbf{F}_{i_2},\cdots,\textbf{F}_{i_s} \}$ and $\Omega_2 = \{ \textbf{F}_{j_1},\textbf{F}_{j_2},\cdots,\textbf{F}_{j_r} \}$
with $1 \leq i_1 < \cdots < i_s \leq 2a$ and $1 \leq j_1 < \cdots < j_r \leq 2a$. Let $\vert \Omega_1 \cap \Omega_2 \vert = p$. Then the product matrix $\textbf{F}_{i_1}\textbf{F}_{i_2}\cdots \textbf{F}_{i_s}$ commutes with $\textbf{F}_{j_1}\textbf{F}_{j_2}\cdots \textbf{F}_{j_r}$ if exactly one of the following is satisfied, and anticommutes otherwise.
\begin{enumerate}
\item $r,s$ and $p$ are all odd.
\item The product $rs$ is even and $p$ is even (including 0).
\end{enumerate}
\end{lemma}
\begin{proof}
When $\textbf{F}_{j_k} \in \Omega_1 \cap \Omega_2$, we note that
\begin{equation*}
(\textbf{F}_{i_1}\textbf{F}_{i_2} \cdots \textbf{F}_{i_s})\textbf{F}_{j_k} = (-1)^{s-1}\textbf{F}_{j_k}(\textbf{F}_{i_1}\textbf{F}_{i_2} \cdots \textbf{F}_{i_s})
\end{equation*}
and when $\textbf{F}_{j_k} \notin \Omega_1 \cap \Omega_2$, we have $(\textbf{F}_{i_1}\textbf{F}_{i_2} \cdots \textbf{F}_{i_s})\textbf{F}_{j_k} = (-1)^{s}\textbf{F}_{j_k}(\textbf{F}_{i_1}\textbf{F}_{i_2} \cdots \textbf{F}_{i_s})$. Now,
\begin{eqnarray*}
(\textbf{F}_{i_1}\textbf{F}_{i_2} \cdots \textbf{F}_{i_s})(\textbf{F}_{j_1}\textbf{F}_{j_2} \cdots \textbf{F}_{j_r})
& = & (-1)^{p(s-1)}(-1)^{(r-p)s}(\textbf{F}_{j_1}\textbf{F}_{j_2} \cdots \textbf{F}_{j_r})(\textbf{F}_{i_1}\textbf{F}_{i_2} \cdots \textbf{F}_{i_s})\\
& = & (-1)^{rs-p}(\textbf{F}_{j_1}\textbf{F}_{j_2} \cdots \textbf{F}_{j_r})(\textbf{F}_{i_1}\textbf{F}_{i_2} \cdots \textbf{F}_{i_s}).
\end{eqnarray*}
\emph{Case} 1) Since $r,s$ and $p$ are all odd, $(-1)^{rs-p}$ = 1.\\
\emph{Case} 2) The product $rs$ is even and  $p$ is even (including 0). Hence $(-1)^{rs-p}$ = 1.
\end{proof}

From Theorem \ref{thm1}, to get a rate-1, 4-group decodable STBC, we need 3 pairwise anticommuting, anti-Hermitian matrices which commute with a group of $2^{a-1}$ Hermitian, pairwise commuting matrices. Once these are identified, the other weight matrices can be easily obtained. From \cite{TiH}, one can obtain $2a$ pairwise anticommuting, anti-Hermitian matrices and the method to obtain these is presented here for completeness. Let
\begin{equation*}
\textbf{P}_1 =\left[\begin{array}{rr}
0 & 1 \\
-1 & 0 \\
\end{array}
\right],
\textbf{P}_2 =\left[\begin{array}{rr}
0 & j \\
j & 0 \\
\end{array}
\right],
\textbf{P}_3 =\left[\begin{array}{rr}
1 & 0 \\
0 & -1 \\
\end{array}
\right]
\end{equation*}
and
$\textbf{A}^{\otimes^{m}} \triangleq \underbrace{\textbf{A}\otimes \textbf{A}\otimes \textbf{A} \cdots \otimes \textbf{A} }_{m~~times  } $. The $2a$ anti-Hermitian, pairwise anti-commuting matrices are
\begin{eqnarray*}
\textbf{F}_1 &= &\pm j \textbf{P}_3^{\otimes^{a}}, \\
\textbf{F}_{2k} &= & \textbf{I}_2^{\otimes^{a-k}} \bigotimes  \textbf{P}_1 \bigotimes \textbf{P}_3^{\otimes^{k-1}}, ~~~k = 1,\cdots,a, \\
\textbf{F}_{2k+1} &= &\textbf{I}_2^{\otimes^{a-k}} \bigotimes \textbf{P}_2 \bigotimes \textbf{P}_3^{\otimes^{k-1}},  ~~~k = 1,\cdots,a-1.
\end{eqnarray*}
\noindent Henceforth, $\textbf{F}_i, i= 1,2,\cdots, 2a$, refer to the matrices obtained using the above method. For a set $\mathcal{S} = \{a_1,a_2,\cdots,a_n\}$, define $\mathbb{P}(\mathcal{S})$ as
\begin{equation*}
\mathbb{P}(\mathcal{S}) \triangleq \left\{a_1^{\lambda_1}a_2^{\lambda_2}\cdots a_n^{\lambda_n}, \lambda_i \in \{0,1\}\right\}.
\end{equation*}

We choose $\textbf{F}_1$, $\textbf{F}_2$ and $\textbf{F}_3$ to be the three pairwise anticommuting, anti-Hermitian matrices (to be placed in the top row along with $\textbf{I}_n$ in Table \ref{table1}. Consider the set $\mathcal{S} = \{j\textbf{F}_4\textbf{F}_5, j\textbf{F}_6\textbf{F}_7,$ $\cdots, j\textbf{F}_{2a-2}\textbf{F}_{2a-1}, \textbf{F}_1\textbf{F}_2\textbf{F}_3 \}$, the cardinality of which is $a-1$. Using Lemma \ref{lemma_2} and Lemma \ref{lemma_3}, one can note that $\mathcal{S}$ consists of pairwise commuting matrices which are Hermitian. Moreover, it is clear that each of the matrices in the set also commutes with $\textbf{F}_1$, $\textbf{F}_2$ and $\textbf{F}_3$. Hence,  $\mathbb{P}(\mathcal{S})$, which has cardinality $2^{a-1}$ is also a set with pairwise commuting, Hermitian matrices which also commute with $\textbf{F}_1$, $\textbf{F}_2$ and $\textbf{F}_3$. The linear independence of $\mathbb{P}(\mathcal{S})$ over $\mathbb{R}$ is easy to see by applying Lemma \ref{lemma_1}. Hence,  we have 3 pairwise anticommuting, anti-Hermitian matrices which commute with a group of $2^{a-1}$ Hermitian, pairwise commuting matrices. Having obtained these, the other weight matrices are obtained from Theorem \ref{thm1}. To illustrate with an example, we consider $n = 8$ and show below how the weight matrices are obtained for the rate-1, 4-group decodable code.

\subsection{An example - $n=8$}
Let $\textbf{F}_i, i=1,2,\cdots,6$ denote the 6 pairwise anticommuting, anti-Hermitian matrices. Choose $\textbf{F}_1$, $\textbf{F}_2$ and $\textbf{F}_3$ to be the three anticommuting matrices required for code construction. Let
\begin{equation*}
\mathcal{S} = \{j\textbf{F}_4\textbf{F}_5, \textbf{F}_1\textbf{F}_2\textbf{F}_3\},~~~
\mathbb{P}(\mathcal{S}) = \{\textbf{I}_8, j\textbf{F}_4\textbf{F}_5, \textbf{F}_1\textbf{F}_2\textbf{F}_3, j\textbf{F}_1\textbf{F}_2\textbf{F}_3\textbf{F}_4\textbf{F}_5\}.
\end{equation*}

The 16 weight matrices of the rate-1, 4-group decodable code for 8 antennas are
as shown in Table \ref{table_eight}. Each column corresponds to the weight matrices belonging to the same group. Note that the product of any two matrices in the first group is some other matrix in the same group.

\begin{table*}
\begin{center}
\begin{tabular}{|c|c|c|c|}
\hline
$\textbf{I}_8$   &  $\textbf{F}_1$   &  $\textbf{F}_2$  & $\textbf{F}_3 $ \\ \hline
$j\textbf{F}_4\textbf{F}_5$   & $j\textbf{F}_1\textbf{F}_4\textbf{F}_5$  & $j\textbf{F}_2\textbf{F}_4\textbf{F}_5$ &  $j\textbf{F}_3\textbf{F}_4\textbf{F}_5$ \\ \hline
$\textbf{F}_1\textbf{F}_2\textbf{F}_3$ & $-\textbf{F}_2\textbf{F}_3$  & $\textbf{F}_1\textbf{F}_3$  & $-\textbf{F}_1\textbf{F}_2$ \\ \hline
$j\textbf{F}_1\textbf{F}_2\textbf{F}_3\textbf{F}_4\textbf{F}_5$   &  $-j\textbf{F}_2\textbf{F}_3\textbf{F}_4\textbf{F}_5$ & $j\textbf{F}_1\textbf{F}_3\textbf{F}_4\textbf{F}_5$  & $-j\textbf{F}_1\textbf{F}_2\textbf{F}_4\textbf{F}_5$ \\ \hline
\end{tabular}
\caption{Weight matrices of a rate-1, 4-group decodable STBC for 8 transmit antennas}
\label{table_eight}
\end{center}
\hrule
\end{table*}

\subsection{Coding gain calculations}\label{sec4A}
Let $\Delta(\textbf{S},\textbf{S}^\prime) \triangleq det\big(\Delta \textbf{S}\Delta \textbf{S}^H\big)$, where $\Delta \textbf{S} \triangleq \textbf{S} - \textbf{S}^\prime, \textbf{S} \neq \textbf{S}^\prime $ denotes the codeword difference matrix. Let $\Delta s_i \triangleq s_i - s_i^{\prime}, i =1,2,\cdots,2n_t$, where $s_i$ and $s_i^{\prime}$ are the real symbols encoding codeword matrices $\textbf{S}$ and $\textbf{S}^\prime$, respectively. Hence,
\begin{eqnarray*}
\Delta(\textbf{S},\textbf{S}^\prime) & = & det\left(\sum_{i=1}^{2n_t}\Delta s_i \textbf{A}_i \sum_{m=1}^{2n_t}\Delta s_m \textbf{A}_m^H\right) =  det\left(\sum_{i=1}^{2n_t} \sum_{m=1}^{2n_t} \Delta s_i \Delta s_m \textbf{A}_i \textbf{A}_m^H\right).
\end{eqnarray*}
Note that because of the nature of construction of the weight matrices, we have
\begin{equation*}
\textbf{A}_i\textbf{A}_m^H = \textbf{A}_{\frac{pn_t}{2}+i}\textbf{A}_{\frac{pn_t}{2}+m}^H, ~~~~ i,m \in \left\{1,2,\cdots,\frac{n_t}{2}\right\},~~~ p \in \{1,2,3\}.
\end{equation*}
Further, since the code is 4-group decodable,
\begin{eqnarray*}
\Delta(\textbf{S},\textbf{S}^\prime) & = & det\left(\sum_{ p=0}^{3}\left(\sum_{i=\frac{pn_t}{2}+1}^{\frac{(p+1)n_t}{2}} \Delta s_i^2\textbf{I}_{n_t} +2 \sum_{i=\frac{pn_t}{2}+1}^{\frac{(p+1)n_t}{2}-1}\sum_{m=i + 1}^{\frac{(p+1)n_t}{2}}\Delta s_i\Delta s_m \textbf{A}_i\textbf{A}_m^H\right)\right).
\end{eqnarray*}
All the weight matrices in the first group are Hermitian and pairwise commuting and the product of any two such matrices is some other matrix in the same group.  It is well known that commuting matrices are simultaneously diagonalizable. Hence,
\begin{equation*}
\textbf{A}_i = \textbf{E}\textbf{D}_i\textbf{E}^H, ~~~ i \in \left\{2,3,\cdots, \frac{n_t}{2} \right\},
\end{equation*}
where $\textbf{D}_i$ is a diagonal matrix. Since $\textbf{A}_i$ is Hermitian as well as unitary, the diagonal elements of $\textbf{D}_i$ are $\pm 1$. The following lemma proves that $\textbf{A}_i$ is traceless.
\begin{lemma}\label{lemma4}
Let $\textbf{F}_i, i = 1,2,\cdots,2a$ be $2^a \times 2^a$ unitary, pairwise anticommuting matrices. Then, the product matrix $\textbf{F}_1^{\lambda_1}\textbf{F}_2^{\lambda_2}\cdots\textbf{F}_{2a}^{\lambda_{2a}}, \lambda_i \in \{0,1\}, i=1,2,\cdots, 2a$, with the exception of $\textbf{I}_{2^a}$, is traceless.
\end{lemma}
\begin{proof}
It is well known that $tr(\textbf{AB})=tr(\textbf{BA})$ for any two matrices $\textbf{A}$ and $\textbf{B}$. Let $\textbf{A}$ and $\textbf{B}$ be two invertible, $n \times n$ anticommuting matrices. Then, $\textbf{A}\textbf{B}\textbf{A}^{-1} = -\textbf{B}$. So,
\begin{eqnarray*}
tr(\textbf{A}\textbf{B}\textbf{A}^{-1}) = -tr(\textbf{B}) \Leftrightarrow tr(\textbf{B}) = -tr(\textbf{B}).
\end{eqnarray*}
\begin{equation}\label{trace}
\therefore tr(\textbf{B}) = 0.
\end{equation}
Similarly, it can be shown that $tr(\textbf{A}) = 0$. By applying Lemma \ref{lemma_3}, it can be seen that any product matrix $\textbf{F}_1^{\lambda_1^\prime}\textbf{F}_2^{\lambda_2^\prime}\cdots \textbf{F}_{2a}^{\lambda_{2a}^\prime}$, excluding $\textbf{I}_{2^a}$,  anticommutes with some other invertible product matrix from the set $\{\textbf{F}_1^{\lambda_1}\textbf{F}_2^{\lambda_2}\cdots\textbf{F}_{2a}^{\lambda_{2a}}, \lambda_i \in \{0,1\}, i = 1,2,3,\cdots,2a\}$. Hence, from \eqref{trace}, we can say that every product matrix $\textbf{F}_1^{\lambda_1}\textbf{F}_2^{\lambda_2}\cdots\textbf{F}_{2a}^{\lambda_{2a}}$ except $\textbf{I}_{2^a}$ is traceless.
\end{proof}

 From the above lemma, $\textbf{A}_i$ except identity is traceless. Hence, $\textbf{D}_i$ has an equal number of '1's and '-1's.
In fact, because of the nature of construction of the matrices $\textbf{F}_i,i=1,2,\cdots,2a$, the product matrices $\textbf{F}_i\textbf{F}_{i+1}$, for even $i$, and the product matrix  $\textbf{F}_1\textbf{F}_{2}\textbf{F}_{3}$ are always diagonal (easily seen from the definition of $\textbf{F}_i$, $i=1,2,\cdots,2a$). Hence, all the weight matrices of the first group excluding $\textbf{A}_1 = \textbf{I}_{n_t}$ are diagonal, with the diagonal elements being $\pm 1 $. Since these diagonal matrices also commute with $\textbf{F}_2$ and $\textbf{F}_3$, the diagonal entries are such that for every odd $i$, if the $(i,i)^{th}$ entry is 1(-1), then, the $(i+1,i+1)^{th}$ entry is also 1(-1, respectively).  To summarize, the properties of $\textbf{A}_i$, $i=2,\cdots, \frac{n_t}{2}$ are listed below.
\begin{eqnarray}
\nonumber
 \textbf{A}_i  =   \textbf{A}_i^H; &~~~& \textbf{A}_i^2  =  \textbf{I}_{n_t}, \\
\nonumber
\textbf{A}_i(m,n)  =  0, ~~ m \neq n; &~~~& \textbf{A}_i(j,j) = \pm 1,~~j=1,2,\cdots, n_t,\\
\label{p5}
tr(\textbf{A}_i) & = & 0, \\
\label{p6}
\textbf{A}_i(j,j) & = & \textbf{A}_i(j+1,j+1), ~~  j = 1, 3, 5,\cdots,n_t-1, 
\end{eqnarray}
\begin{eqnarray}
\label{p7}
\textbf{A}_i \textbf{A}_j & = & \textbf{A}_k, ~~~~ i,j,k \in \left\{1,2,\cdots,\frac{n_t}{2}\right\}.
\end{eqnarray}

In view of these properties,
\begin{equation*}
\Delta(\textbf{S},\textbf{S}^\prime) =  det\left(\sum_{p=0}^{3}\left(\sum_{i=\frac{pn_t}{2}+1}^{\frac{(p+1)n_t}{2}} \Delta s_i^2\textbf{I}_{n_t} + 2\sum_{i=\frac{pn_t}{2}+1}^{\frac{(p+1)n_t}{2}-1}\sum_{m=i + 1}^{\frac{(p+1)n_t}{2}}\Delta s_i \Delta s_m \textbf{D}_{im}\right) \right),
\end{equation*}
where $\textbf{D}_{im} = \textbf{A}_i\textbf{A}_m =  \textbf{A}_k$ for some $k \in \{1,2,\cdots,\frac{n_t}{2} \}$, and
\begin{equation}\label{prod_sum}
\Delta(\textbf{S},\textbf{S}^\prime) =  \prod_{j=1}^{n_t}\sum_{p=0}^{3}\left( \sum_{i=1}^{\frac{n_t}{2}}d_{ij} \Delta s_{\frac{pn_t}{2}+i}\right)^2,
\end{equation}
where $d_{ij} = \pm 1$ and $d_{1j} = 1$. In fact, $d_{ij} = \textbf{A}_i(j,j)$, $i = 1,2, 3, \cdots, \frac{n_t}{2}$. From \eqref{prod_sum}, $\Delta(\textbf{S},\textbf{S}^\prime)$ is a product of the sum of squares and it is minimized when only one group, say $p=0$, gives a non-zero contribution. Hence,
\begin{equation*}
\min_{\textbf{S},\textbf{S}^\prime}(\Delta(\textbf{S},\textbf{S}^\prime))  = \min_{\Delta s_i} \left(\prod_{j=1}^{n_t}\left(\sum_{i=1}^{\frac{n_t}{2}}d_{ij} \Delta s_i\right)^2\right),
\end{equation*}
where $\underset{x}{\operatorname{min}}(y)$ denotes the minimum value of $y$ over all possible values of $x$.
From \eqref{p6},
\begin{equation}\label{tem}
\min_{\textbf{S},\textbf{S}^\prime}(\Delta(\textbf{S},\textbf{S}^\prime))  =  \min_{\Delta s_i}\left(\prod_{j=1}^{\frac{n_t}{2}}\left(\sum_{i=1}^{\frac{n_t}{2}}d_{i(2j-1)} \Delta s_i\right)^4\right).
\end{equation}
We need the minimum determinant to be as high a non-zero number as possible. In this regard, let
\begin{equation}\label{wmatrix}
\textbf{W} \triangleq \sqrt{\frac{2}{n_t}}[w_{ij}], ~~  w_{ij} = d_{i(2j-1)}, ~~ i,j = 1,2,\cdots, \frac{n_t}{2}
\end{equation}
and
\begin{equation*}
\textbf{y}_p \triangleq [y_{_{\frac{n_tp}{2}+1}},y_{_{\frac{n_tp}{2}+2}},\cdots, y_{_{\frac{n_t(p+1)}{2}}}]^T = \textbf{W} [s_{_{\frac{n_tp}{2}+1}},s_{_{\frac{n_tp}{2}+2}},\cdots, s_{_{\frac{n_t(p+1)}{2}}}]^T, ~~ p = 0,1,2,3.
\end{equation*}

\begin{lemma}
$\textbf{W}$ as defined in \eqref{wmatrix} is an orthogonal matrix.
\end{lemma}
\begin{proof}
From \eqref{wmatrix}, it can be noted that the columns of $\textbf{W}$ are obtained from the diagonal elements of $\textbf{A}_i$, $i=1,2,\cdots,\frac{n_t}{2}$. Each element of a column $i$ of $\textbf{W}$ corresponds to every odd numbered diagonal element of $\textbf{A}_i$. Denote the $i^{th}$ column of $\textbf{W}$ by $\textbf{w}_i$. Applying \eqref{p6}, \eqref{p7} and \eqref{p5} in that order,
\begin{eqnarray*}
\langle \textbf{w}_i,\textbf{w}_j \rangle = \frac{1}{n_t}tr(\textbf{A}_i\textbf{A}_j) = \frac{1}{n_t} tr(\textbf{A}_k)= \delta_{ij},
\end{eqnarray*}
where
\begin{equation*}
\delta_{ij} = \left\{ \begin{array}{ccc}
 0, & \textrm{if} & i \neq j,\\
 1, & \textrm{otherwise}. & \\
\end{array} \right.
\end{equation*}
Hence, $\textbf{W}$ is orthogonal.
\end{proof}

Substituting $\textbf{y}_p$ in \eqref{tem}, we get
\begin{equation*}
\min_{\textbf{S},\textbf{S}^\prime}(\Delta(\textbf{S},\textbf{S}^\prime))  =  \min_{\textbf{y}_0} \left(\prod_{j=1}^{\frac{n_t}{2}} y_j^4\right).
\end{equation*}

So, the minimum determinant is a power of the minimum product distance in $n_t/2 $ real dimensions. If $\textbf{y}_p \in \mathbb{Z}^{\frac{n_t}{2}}$, the product distance can be maximized by premultiplying $\textbf{y}_p$ with a suitable orthogonal rotation matrix $\textbf{V}$ given in \cite{full_div_rot}.  This operation maximizes the minimum determinant and hence the coding gain. So, the $2n_t$ real symbols of the rate-1, 4-group decodable code are encoded by grouping $\frac{n_t}{2}$ real symbols each into 4 groups and each group of symbols taking value from a unitarily rotated vector belonging to $\mathbb{Z}^{\frac{n_t}{2}}$, the rotation matrix being $\textbf{W}^T\textbf{V}$. For 4 transmit antennas,
\begin{equation*}
\textbf{W} = \frac{1}{\sqrt{2}}\left[ \begin{array}{rr}
1 & -1  \\
1 &  1 \\
\end{array}\right], ~~~
\textbf{V} = \left[ \begin{array}{rr}
 0.8507 &  -0.5257 \\
 0.5257  &  0.8507\\
\end{array}\right],
\end{equation*}
and for 8 transmit antennas,
\begin{equation*}
\textbf{W} = \frac{1}{2}\left[ \begin{array}{rrrr}
1 & -1 & -1 & 1 \\
1 &  1 &  1 & 1\\
1 & -1 &  1 & -1\\
1 &  1 & -1 & -1 \\
\end{array}\right], ~~~
\textbf{V} = \left[ \begin{array}{rrrr}
-0.3664 & -0.7677 &  0.4231  & 0.3121 \\
-0.2264 & -0.4745 & -0.6846 & -0.5050 \\
-0.4745 &  0.2264 & -0.5050  & 0.6846 \\
-0.7677 & 0.3664 & 0.3121 & -0.4231 \\
\end{array}\right].
\end{equation*}

If the practically used square QAM constellation of size $M$ is used, encoding is done as follows : the $n_t$ complex symbols in each codeword matrix take values from the $M$-QAM and are split into two groups, one group consisting of the real parts of the $n_t$ symbols and the other group consisting of the imaginary parts. Each group is further divided into two subgroups, each consisting of $n_t/2$ real symbols. So, in all, there are 4 groups consisting of $n_t/2$ real symbols. As used before, denoting the column vectors consisting of the symbols in a group by $\textbf{y}_p$, $p = 0,1,2,3$ (the entries of $\textbf{y}_p$  take values independently from $\sqrt{M}$-PAM), let $\textbf{s}_p = \textbf{W}^T\textbf{V}\textbf{y}_p$, where $\textbf{W}$ and $\textbf{V}$ are as explained before. Then the codeword matrix is given by
\begin{equation*}
\textbf{S} =  \sum_{p=0}^{3} \sum_{i=1}^{\frac{n_t}{2}} s_{_{\frac{pn_t}{2}+i}}\textbf{A}_{\frac{pn_t}{2} + i}.
\end{equation*}
Note that the above codeword matrix can also be expressed as 
\begin{equation}\label{new_weight}
\textbf{S} =  \sum_{p=0}^{3} \sum_{i=1}^{\frac{n_t}{2}} y_{_{\frac{pn_t}{2}+i}}\textbf{A}_{\frac{pn_t}{2} + i}^\prime,
\end{equation}
where $\textbf{A}_{\frac{pn_t}{2} + i}^\prime = \sum_{j=1}^{\frac{n_t}{2}}\omega_{ji}\textbf{A}_{\frac{pn_t}{2} +j}$, $p=0,1,2,3$, with $\omega_{ji}$ being the $(j,i)^{th}$ element of $\textbf{W}^T\textbf{V}$. Clearly, the weight matrices $\textbf{A}_{\frac{pn_t}{2} + i}^\prime$, $p=0,1,2,3$, satisfy the condition
\begin{center}
 $\textbf{A}_{\frac{ln_t}{2} + i}^\prime \left(\textbf{A}_{\frac{mn_t}{2} + j}^\prime\right)^H + \textbf{A}_{\frac{mn_t}{2} + j}^\prime \left(\textbf{A}_{\frac{ln_t}{2} + i}^\prime \right)^H = \textbf{O}_{n_t}$,  for $0 \leq l < p \leq 3$, and $i,j = 0,1, \cdots, \frac{n_t}{2}$. 
\end{center}
Consequently, the ML-decoding complexity of the code is of the order of $M^{\frac{n_t-2}{4}}$. This is because there are four groups consisting of $n_t/2$ real symbols each and the symbols in each group can be decoded independently from the symbols in the other groups. In decoding the symbols in the same group jointly, one needs to make a search over $\sqrt{M}^{\frac{n_t}{2}} = M^{\frac{n_t}{4}}$ possibilities for the symbols, since the real and the imaginary parts of a signal point in a square $M$-QAM have only $\sqrt{M}$ possible values each (the real and the imaginary parts of a signal point of a square $M$-QAM take values from a $\sqrt{M}$-PAM constellation). However, one need not make an exhaustive search over all the possible $M^{\frac{n_t}{4}}$ values for the $n_t/2$ symbols. For every possible value of the first $\frac{n_t}{2} -1$ real symbols, the last symbol is evaluated by {\it quantization} \cite{SrR_arxiv}. Hence, the worst case ML-decoding complexity is of the order of $\sqrt{M}^{\frac{n_t}{2} -1} = M^{\frac{n_t-2}{4}} $ only. Fig. \ref{fig_4gp} gives a comparison of the symbol error rate for the proposed STBC, the 4-group decodable STBC proposed by Yuen et al. \cite{4gp1} and the 4-group-decodable STBC proposed by Rajan \cite{4gp2}, all for the $8\times 1$ MIMO system. The plots reveal that all the STBCs have the same performance for QAM constellations. Independently, we have computed that all the three codes have the same minimum determinant for QAM constellations.

\section{Extension to higher number of receive antennas}\label{sec5}
When $n_r = 1$, a rate-1, 4-group decodable STBC is the best full-rate STBC possible in terms of ML-decoding complexity and as a result, ergodic capacity. However, when $n_r > 1$, we need more weight matrices to meet the full-rate criterion. In literature, there does not exist a 4-group decodable STBC with rate greater than 1. So, it is unlikely, though not proven, that there exists a full-rate, multi-group ML-decodable STBC with full-diversity for $n_r > 1$. So, for $n_r >1$, we relax the requirement of multi-group decodability and simply aim for some reduction in the ML-decoding complexity and self-interference. Let $n_t = 2^a$. We know that if $\textbf{F}_i,i=1,2,\cdots,2a$ are pairwise anticommuting, invertible matrices, then, the set $\mathcal{F} \triangleq  \{\textbf{F}_1^{\lambda_1}\textbf{F}_2^{\lambda_2}\cdots\textbf{F}_{2a}^{\lambda_{2a}}$, with $\lambda_i \in \{0,1\}, i = 1,2,\cdots,2a\}$ is linearly independent over $\mathbb{C}$. Hence, the set $\mathcal{M} = \{ \mathcal{F}, j\mathcal{F} \}$ is linearly independent over $\mathbb{R}$. As a result, the elements of $\mathcal{M}$ can be used as weight matrices of a full-rate STBC for $n_r > 1$. Keeping in view that the self-interference has to be minimized, it is important to choose the weight matrices judiciously. The idea is that given a full-rate STBC for $n_r-1$ receive antennas, obtain the additional weight matrices of a full-rate STBC for $n_r$ receive antennas by using the weight matrices of a rate-1, 4-group decodable STBC such that after the addition of the new weight matrices, the set consisting of the weight matrices of the rate-$n_r$ code is linearly independent over $\mathbb{R}$. This is achieved as follows.
\begin{enumerate}
\item Obtain a rate-1, 4-group decodable STBC by using the construction detailed in Section \ref{sec4}. Due to the nature of the construction, the product of any two weight matrices is always some other weight matrix of the code, up to negation. Denote the set of weight matrices by $\mathcal{G}_1$. %
\item From the set $\mathcal{F}$, choose a matrix that does not belong to $\mathcal{G}_1$ and multiply it with the elements of $\mathcal{G}_1$ to obtain a new set of weight matrices, denoted by $\mathcal{G}_2$. Clearly, the two sets will not have any matrix in common. To see this, let $\textbf{A} \in \mathcal{G}_1$ and $\textbf{B} \in \mathcal{F} \bigcap (\mathcal{M}\setminus \mathcal{G}_1)$, where $\textbf{B}$ is the matrix chosen to be multiplied with the elements of $\mathcal{G}_1$. Let $\textbf{BA} = \textbf{C} \in \mathcal{G}_1$. Hence, $\textbf{B} = \textbf{C}\textbf{A}^H = \pm \textbf{CA}$ and $\textbf{CA}$ belongs to $\mathcal{G}_1$, up to negation. This contradicts the fact that $\textbf{B} \in \mathcal{F} \bigcap (\mathcal{M}\setminus \mathcal{G}_1)$. So, $\textbf{C}$ cannot belong to $\mathcal{G}_1$.

The weight matrices of $\mathcal{G}_2$ form a new, rate-1, 4-group decodable STBC. This is because the ML-decoding complexity does not change by multiplying the weight matrices of a code with a unitary matrix. In this case, we have multiplied the elements of $\mathcal{G}_1$ with an element of $\mathcal{F}$, which is a unitary matrix. Now, $\mathcal{G}_1 \bigcup \mathcal{G}_2$ is the set of weight matrices of a rate-2 code with an ML-decoding complexity of $M^{n_t}.M^{\frac{n_t-2}{4}} = M^{\frac{5n_t-2}{4}}$. This is achieved by decoding the last $n_t$ symbols with a complexity of $M^{n_t}$ and then conditionally decoding the first $n_t$ symbols using the 4-group decodability property as explained in Section \ref{sec4A}.
\item For increasing $n_r$, repeat as in the second step, obtaining new rate-1, 4-group decodable codes and then appending their weight matrices to obtain a new, rate-$n_r$ code with an ML-decoding complexity of  $M^{n_t(n_r-\frac{3}{4})-0.5}$. The new set of weight matrices is $\bigcup_{i=1}^{n_r}\mathcal{G}_{i}$.
\item When all the elements of $\mathcal{F}$ have been exhausted (this occurs when $n_r = n_t/2$), Step 3 can be continued till $n_r = n_t$ by choosing the  matrices that are to be multiplied with the elements of $\mathcal{G}_1$ from $j\mathcal{F}\bigcap (\mathcal{M}\setminus \bigcup_{i=1}^{n_r-1}\mathcal{G}_{i})$. Note from Lemma \ref{lemma_1} that this does not spoil the linear independence of the weight matrices over $\mathbb{R}$.
\end{enumerate}

{\it Note} : In the case of the Perfect codes for $n_t$ transmit antennas, a layer \cite{ORBV}, \cite{new_per} corresponds to $n_t$ complex symbols. In case of our generalized Silver codes, a layer corresponds to a rate-1, 4-group decodable code encoding $n_t$ complex symbols. Also, the Silver code for an $n_t \times n_r$ system refers to the STBC containing $n_{min} = min(n_t,n_r)$ individual rate-1, 4-group decodable codes, a property due to which self-interference is greatly reduced compared with other known full-rate codes.

\subsection{An illustration for $n_t = 4$}\label{illustration}
For $n_t =4$, let $\textbf{F}_1, \textbf{F}_2, \textbf{F}_3$ and $\textbf{F}_4$ be the four anticommuting, anti-Hermitian matrices obtained by the method presented in \cite{TiH}. Let $\mathcal{F} = \{\textbf{F}_1^{\lambda_1}\textbf{F}_2^{\lambda_2} \textbf{F}_3^{\lambda_3} \textbf{F}_4^{\lambda_4}, \lambda_i \in \{0,1\}, i = 1,2,3,4\}$. The rate-1, 4-group decodable code has the following 8 weight matrices, with weight matrices in each column belonging to the same group:
\begin{center}
\begin{tabular}{|c|c|c|c|}
\hline
$\textbf{I}_4$   &  $\textbf{F}_1$   &  $\textbf{F}_2$  & $\textbf{F}_3 $ \\ \hline
$\textbf{F}_1\textbf{F}_2\textbf{F}_3$ & $-\textbf{F}_2\textbf{F}_3$  & $\textbf{F}_1\textbf{F}_3$  & $-\textbf{F}_1\textbf{F}_2$ \\ \hline
\end{tabular}
\end{center}

Hence, $\mathcal{G}_1 = \{ \textbf{I}_4,\textbf{F}_1,\textbf{F}_2,\textbf{F}_3,\textbf{F}_1\textbf{F}_2\textbf{F}_3,-\textbf{F}_2\textbf{F}_3,\textbf{F}_1\textbf{F}_3,-\textbf{F}_1\textbf{F}_2\}$, Now, we choose a matrix from $\mathcal{F}$ which does not belong to $\mathcal{G}_1$. One such matrix is $\textbf{F}_4$. Pre-multiplying all the elements of $\mathcal{G}_1$ with $\textbf{F}_1$ and applying the anticommuting property, we obtain a new rate-1, 4-group decodable code, whose weight matrices are as follows:
\begin{center}
\begin{tabular}{|c|c|c|c|}
\hline
$\textbf{F}_4$   &  $-\textbf{F}_1\textbf{F}_4$   &  $-\textbf{F}_2\textbf{F}_4$  & $-\textbf{F}_3\textbf{F}_4 $ \\ \hline
$-\textbf{F}_1\textbf{F}_2\textbf{F}_3\textbf{F}_4$ & $-\textbf{F}_2\textbf{F}_3\textbf{F}_4$  & $\textbf{F}_1\textbf{F}_3\textbf{F}_4$  & $-\textbf{F}_1\textbf{F}_2\textbf{F}_4$ \\ \hline
\end{tabular}
\end{center}
Hence, $\mathcal{G}_2 = \textbf{F}_4\mathcal{G}_1 = \{\textbf{F}_4,-\textbf{F}_1\textbf{F}_4,$ $-\textbf{F}_2\textbf{F}_4,-\textbf{F}_3\textbf{F}_4,-\textbf{F}_1\textbf{F}_2\textbf{F}_3\textbf{F}_4,-\textbf{F}_2\textbf{F}_3\textbf{F}_4, \textbf{F}_1\textbf{F}_3\textbf{F}_4,-\textbf{F}_1\textbf{F}_2\textbf{F}_4\}$ and $\mathcal{G}_1 $ $\bigcup \mathcal{G}_2$ is the set of weight matrices of the rate-2 STBC, which is full rate with an ML-decoding complexity of the order of $M^{4.5}$.

Now, since there are no more elements left in $\mathcal{F}$ (neglecting negation), we can choose elements from $j\mathcal{F}$. To construct a rate-3 code for 3 transmit antennas, we multiply the elements of $\mathcal{G}_1$ by $j\textbf{I}_4$ to obtain the set $\mathcal{G}_3 = j\mathcal{G}_1$. The weight matrices of the rate-3 code constitute the set $\mathcal{G}_1 \bigcup \mathcal{G}_2 \bigcup \mathcal{G}_3$. Similarly, the weight matrices of a full-rate code for $n_r \geq 4$ are the elements of the set $ \mathcal{G}_1 \bigcup \mathcal{G}_2 \bigcup \mathcal{G}_3  \bigcup \mathcal{G}_4$, where $\mathcal{G}_4 = j\textbf{F}_4\mathcal{G}_1 = j\mathcal{G}_2$. It is obvious that $\mathcal{G}_1$, $\mathcal{G}_2$, $\mathcal{G}_3$ and $\mathcal{G}_4$ represent the weight matrices of four individual rate-1, 4-group decodable codes, respectively.

\subsection{Structure of the $\textbf{R}$-matrix and ML-decoding complexity}
The popular sphere decoding \cite{sphere_decoding} technique is used to perform the ML-decoding of linear dispersion STBCs utilizing lattice constellations. A QR-decomposition of $\textbf{H}_{eq}$, the equivalent channel matrix, is performed to obtain $\textbf{H}_{eq} = \textbf{QR}$ and the ML-decoding metric is given by
\begin{equation*}
 \textbf{M}\left(\textbf{s}\right) = \left \Vert \widetilde{vec(\textbf{Y})} - \sqrt{\frac{SNR}{n_t}}\textbf{H}_{eq}\textbf{s} \right\Vert^2  = \left \Vert \textbf{y}^\prime - \sqrt{\frac{SNR}{n_t}}\textbf{R}\textbf{s}  \right \Vert^2,
\end{equation*}
where $\textbf{y}^\prime = \textbf{Q}^T\widetilde{vec(\textbf{Y})}$. The $\textbf{R}$-matrix of the Silver code for the $n_t \times n_r$ system has the following structure, irrespective of the channel realization:
\begin{equation*}\label{rmatrix}
 \textbf{R} =  \left[\begin{array}{cccc}
\textbf{D} &  \textbf{X} &  \ldots &  \textbf{X} \\
\textbf{O}_{2n_t} &  \textbf{D} &  \ldots &  \textbf{X} \\
\vdots &   \ddots  & \ddots & \vdots \\
\textbf{O}_{2n_t} & \textbf{O}_{2n_t} & \ldots & \textbf{D} \\
\end{array} \right]
\end{equation*}
where $\textbf{X} \in \mathbb{R}^{2n_t \times 2n_t}$ is a random non-sparse matrix whose entries depend on the channel coefficients and $\textbf{D}  = \textbf{I}_4 \otimes \textbf{T} $, with $\textbf{T} \in \mathbb{R}^{\frac{n_t}{2} \times \frac{n_t}{2}}$ being an upper triangular matrix. The reason for this structure is that the weight matrices of the Silver code for an $n_t \times n_r$ system are also the weight matrices of $min(n_t,n_r)$ separate rate-1, 4-group decodable codes (as illustrated in Sec. \ref{sec5}). As a result of the structure of $\textbf{D}$, the $\textbf{R}$-matrix has a large number of zeros in the upper block, and hence, compared to other existing codes, the generalized Silver codes have lower average ML-decoding complexity. The worst case ML-decoding complexity is of the order of $(M^{n_t(n_{min}-1)})(M^{\frac{n_t-2}{4}}) = M^{n_t(n_{min}-\frac{3}{4})-0.5}$, which is because in decoding the symbols, a search is to be made over all possible values of the last $n_t(n_{min}-1)$ complex symbols (which requires a complexity of the order of $M^{n_t(n_{min}-1)}$), while the remaining $n_t$ symbols can be {\it conditionally decoded} with a complexity of $M^{\frac{n_t-2}{4}}$ only, once the last $n_t(n_{min}-1)$ symbols are fixed (a detailed explanation on conditional ML-decoding has been presented in \cite{BHV}, \cite{SrR_arxiv}). In simple words, to decode the Silver code, one does not need a $2n_tn_{min}$ dimensional real sphere decoder. All one requires is a $2n_t(n_{min}-1)$ dimensional real sphere decoder in conjunction with four parallel $(n_t-2)/2$ dimensional real sphere decoders. The decrease in the ML-decoding complexity is evident from the decrease in the dimension of the real sphere decoder from $2n_tn_{min}$ to $2n_t(n_{min}-1) +\frac{n_t-2}{2}$. 

\subsection{Information Losslessness for $n_r \geq n_t$}
 For $n_r \geq n_t$, the Silver code is information lossless because its normalized generator matrix (normalization is done to ensure an appropriate $SNR$ at each receive antenna) is orthogonal. To see this, the generator matrix for $n_r \geq n_t$ is given as
\begin{equation*}
\textbf{G} = \frac{1}{\sqrt{n_t}} [\widetilde{vec(\textbf{A}_1)}\ \widetilde{vec(\textbf{A}_2)}\ \cdots \ \widetilde{vec(\textbf{A}_{2n_t^2})}],
\end{equation*}
where $\textbf{A}_i \in \mathcal{M}, i=1,2,\cdots,2n_t^2$, are the weight matrices obtained as mentioned in Sec. \ref{sec5}, with $\mathcal{M} = \{\mathcal{F}, j\mathcal{F}\}$, where $\mathcal{F} = \{\textbf{F}_1^{\lambda_1}\textbf{F}_2^{\lambda_2}\cdots\textbf{F}_{2a}^{\lambda_{2a}}, \lambda_i \in \{0,1\}, i = 1,2,3,\cdots,2a\}$. For $i,j \in\{1,2,\cdots, 2n_t^2\}$, we have
\begin{eqnarray}
\label{col_orth}
 \langle \widetilde{vec(\textbf{A}_i)}, \widetilde{vec(\textbf{A}_j)} \rangle &  = & {\rm{real}}\left(tr\left(\textbf{A}_i^H\textbf{A}_j\right)\right)\\
\label{gen1}
& = &  \pm {\rm{real}}\left(tr(\textbf{A}_i\textbf{A}_j)\right) \\
\nonumber
& = & \left\{ \begin{array}{ll}
 {\rm{real}}\left(tr(\textbf{I}_{n_t})\right) & {\rm{if}}~ i=j\\
 {\rm{real}}\left(tr(j\textbf{I}_{n_t})\right) & {\rm{if}} ~\textbf{A}_i = j\textbf{A}_j\\
\pm {\rm{real}}\left(tr(\textbf{A}_k)\right) &  {\rm{otherwise, where}}~ \pm \textbf{A}_k \in \mathcal{M}\setminus\{\textbf{I}_{n_t},j \textbf{I}_{n_t}\}\\
\end{array}\right.\\
\label{gen2}
& = & n_t\delta_{ij}.
\end{eqnarray}
Equation \eqref{gen1} holds because $\textbf{A}_i$, $i = 1,\cdots, 2n_t^2$ are either Hermitian or anti-Hermitian, and \eqref{gen2} follows from Lemma \ref{lemma4}.

\begin{lemma}\label{lemma_trace}
 Tracelessness of the self-interference matrix is equivalent to column orthogonality of the generator matrix.
\end{lemma}
\begin{proof}
 Using the definition of the self-interference matrix $\textbf{S}^{int}$, given in Definition \ref{def8}, 
\begin{eqnarray}
\nonumber
 tr\left(\textbf{S}^{int}\right)& = &\sum_{i=1}^{2k-1} \sum_{j>i}^{2k}s_is_j tr\left[\left(\textbf{A}_i\textbf{A}_j^H + \textbf{A}_j\textbf{A}_i^H \right)\right]  = 2\sum_{i=1}^{2k-1} \sum_{j>i}^{2k}s_is_j\left( {\rm{real}}\left[tr\left(\textbf{A}_i\textbf{A}_j^H\right)\right]\right) \\
\label{line1}
& = &2\sum_{i=1}^{2k-1} \sum_{j>i}^{2k}s_is_j \langle \widetilde{vec(\textbf{A}_i)}, \widetilde{vec(\textbf{A}_j)} \rangle,
\end{eqnarray}
where \eqref{line1} follows from \eqref{col_orth}. From \eqref{line1}, it is clear that column orthogonality of the generator matrix is equivalent to tracelessness of the self-interference matrix.
\end{proof}

Recall that the second criterion  given to maximize $I_2$ (given by \eqref{coefficient_2}) requires that $\textbf{S}_{ij} = \textbf{A}_i\textbf{A}_j^H +  \textbf{A}_j\textbf{A}_i^H$, $i \neq j$, be traceless. It is clear from Lemma \ref{lemma_trace} that for our STBCs, $\textbf{S}_{ij}$ is traceless for $i \neq j$.

\subsection{The Silver code for two transmit antennas}

The Silver code \cite{tirk_combined}, \cite{PGA} for two antennas, which is well known for being a low complexity, full-rate, full-diversity STBC for $n_r \geq 2$, transmits 2 complex symbols per channel use. A codeword matrix of the Silver code is given as
\begin{equation*}
\textbf{S} =  \left[ \begin{array}{rr}
s_1+js_2   &  s_3+js_4  \\
-s_3+js_4  & s_1 -js_2  \\
\end{array}\right] + j\left[ \begin{array}{rr}
s_5+js_6   &  s_7+js_8  \\
-s_7+js_8  & s_5 -js_6  \\
\end{array}\right]\textbf{U},
\end{equation*}
where
\begin{equation*}
\textbf{U} = \frac{1}{\sqrt{7}}\left[\begin{array}{cc}
1 + j & 1 + 2j \\
-1+2j &  1 - j \\
\end{array}\right ].
\end{equation*}

The codeword encodes 8 real symbols $s_1,s_2,\cdots,s_8$, each taking values independently from a regular $\sqrt{M}$-PAM constellation. The first four weight matrices are that of the Alamouti code, given by
\begin{equation*}
\textbf{A}_1 =  \left[ \begin{array}{rr}
1   &  0  \\
0  & 1 \\
\end{array}\right], ~ \textbf{A}_2 =  \left[ \begin{array}{rr}
j   &  0  \\
0  & -j \\
\end{array}\right], ~ \textbf{A}_3 =  \left[ \begin{array}{rr}
0   &  1  \\
-1  & 0\\
\end{array}\right], ~ \textbf{A}_4 =  \left[ \begin{array}{rr}
0   &  j  \\
j & 0 \\
\end{array}\right].
\end{equation*}
Note that the Alamouti code is 4-group decodable for 2 transmit antennas. The Silver code's next 4 weight matrices are obtained by multiplying the first four weight matrices by $j$. To make the code achieve full-diversity with the highest possible coding gain, post-multiplication by $\textbf{U}$ is performed. It can be checked that $\textbf{U} = \frac{1}{\sqrt{7}}(\textbf{A}_1+\textbf{A}_2+\textbf{A}_3+2\textbf{A}_4)$. Effectively, the last 4 weight matrices of the silver code are $j\textbf{A}_i\textbf{U}$, $i=1,\cdots,4$, which also form another rate-1, 4-group decodable code. The unitary matrix $\textbf{U}$ is so cleverly chosen that in addition to providing full-diversity with a high coding gain, the generator matrix is orthogonal (which can be checked using \eqref{gen1}), making the code information lossless for $n_r \geq 2$. The Silver code compares very well with the well known Golden code in error performance, while offering lower ML-decoding complexity of the order of $M^2$.

\subsection{Achievability of Full-diversity}
The following theorem, (Theorem I, \cite{lakshmi}) guarantees that full-diversity is possible for the generalized Silver codes with the real symbols taking values from PAM constellations, denoted by $\mathcal{A}_{PAM}$.
\begin{theorem}\label{thm3}
 For any given $n \times n$ square linear design $\mathcal{S} \triangleq \left\{\textbf{S} = \sum_{i=1}^{k}s_i \textbf{A}_i ~ \vert~ s_i \in \mathcal{A}_{PAM} ,  ~i = 1,2, \right.$ $\left.\cdots,k\right\}$, encoding $k$ real symbols with full-rank weight matrices $\textbf{A}_i$, there exist $\alpha_i \in \mathbb{C}$, $i = 1,\cdots, k$, such that the STBC $\mathcal{S}^\prime \triangleq \left\{ \textbf{S} = \sum_{i=1}^{k}s_i\alpha_i\textbf{A}_i ~ \vert ~ s_i \in \mathcal{A}_{PAM}, ~ i = 1,2,\cdots,k  \right\}$ offers full diversity.
\end{theorem}

Since all the weight matrices of the generalized Silver code are either Hermitian or anti-Hermitian and hence full-ranked, full-diversity is achievable with the generalized Silver codes. However, finding out explicitly the values of $\alpha_i$ is an open problem. For the full-rate codes for 1 receive antenna, in Section \ref{sec4A}, we have identified the encoding scheme which not only provides full-diversity, but also maximizes the coding gain for PAM constellations. For the generalized Silver codes for higher number of receive antennas, each layer, corresponding to a rate-1, 4-group decodable code, is encoded as explained in \ref{sec4A}. Note from \eqref{new_weight} that this type of encoding neither reduces the number of matrix pairs satisfying Hurwitz-Radon orthogonality nor spoils the column orthogonality of the Generator matrix. In addition, we use a certain scaling factor to be multiplied with a certain subset of weight matrices to enhance the coding gain. The choice of the scaling factor is based on computer search. With the use of the scaling factor, the generalized Silver codes perform very well when compared with the punctured Perfect codes. Although we cannot mathematically prove that our codes have full-diversity with the constellation that we have used for simulation, the simulation plots seem to suggest that our codes have full-diversity, since the error performance of our codes matches that of the comparable punctured Perfect codes, which have been known to have full-diversity.

\section{Simulation results}\label{sec6}
In all the simulation scenarios in this section, we consider the Rayleigh block fading MIMO channel.
\subsection{4 Tx}
We consider three MIMO systems - $4 \times 2$, $4 \times 3$ and $4\times4$ systems. The codes are constructed as illustrated in Subsection  \ref{illustration}. To enhance the performance of our code for the $4 \times 2$ system, we have multiplied the weight matrices of $\mathcal{G}_2$ (as defined in Subsection  \ref{illustration}) with the scalar $e^{j\pi/4}$. This is done primarily to enhance the coding gain, which was observed to be the highest when the scalar $e^{j\pi/4}$ was chosen. It is to be noted that this action does not alter the ML-decoding complexity and the column orthogonality of the generator matrix (so, the resultant weight matrices still satisfy the tracelessness criterion). Consequently, the weight matrices of the Silver code for the $4 \times 2$ system can be viewed to be from $\mathcal{G}_1 \bigcup e^{j\pi/4} \mathcal{G}_2$. For the $4 \times 3$ MIMO system, the weight matrices of the Silver code are from the set $\mathcal{G}_1 \bigcup e^{j\pi/4}\mathcal{G}_2 \bigcup j\mathcal{G}_1$, while the weight matrices of the Silver code for the $4 \times 4 $ system are from the set  $\mathcal{G}_1 \bigcup e^{j\pi/4}\mathcal{G}_2 \bigcup j\mathcal{G}_1 \bigcup je^{j\pi/4}\mathcal{G}_2 $. Fig. \ref{fig_cap} shows the plot of the maximum mutual information achievable with our codes and the punctured Perfect codes \cite{ORBV} for $4\times2$ and $4 \times 3$ systems. In both the cases, our codes allow higher mutual information than the punctured Perfect code, as was expected. Regarding error performance, we have chosen 4 QAM for our simulations and encoding is done as explained in Subsection \ref{sec4A}.
\begin{enumerate}
\item $4 \times 2$ MIMO \\
Fig. \ref{fig1} shows the plots of the symbol error rate (SER) as a function of the SNR at each receive antenna for five codes - the DjABBA code \cite{HTW}, the punctured Perfect code for 4 transmit antennas, the Silver code for the $4 \times 2$ system, the EAST code \cite{barry} and Oggier's code from crossed product Algebra with a provable NVD property \cite{Oggier_fast}. Since the number of degrees of freedom of the channel is only 2, we use the Perfect code with 2 of its 4 layers punctured. Our code and the EAST code have the best performance. It is to be noted that the curves for the Silver code for the $4 \times 2$ system and the EAST code coincide. Also, the Silver code for the $4 \times 2$ system is the same as the one presented in \cite{SrR_arxiv}, but has been designed using a new, systematic method. The Silver code for the $4 \times 2$ system and the EAST code have an ML-decoding complexity of the order of $M^{4.5}$ for square QAM constellation, while the DjABBA and Oggier's code have an ML-decoding complexity of order $M^{6}$ and $M^{5.5}$, respectively.
\item $4 \times 3$ MIMO \\
Fig. \ref{fig2} shows the plots of the SER as a function of the SNR at each receive antenna for two codes - the punctured perfect code (puncturing one of its 4 layers) and the Silver code for the $4 \times 3$ system.  The Silver code for the $4 \times 3$ system has a marginally better performance than the punctured perfect code in the low to medium SNR range. It has an ML-decoding complexity of the order of $M^{8.5}$ while that of the punctured Perfect code is $M^{11}$ (this reduction from $M^{12}$ to $M^{11}$ is due to the fact that the real and the imaginary parts of the last symbol can be evaluated by quantization, once the remaining symbols have been fixed).
\item $4 \times 4$ MIMO \\
Fig. \ref{fig3} shows the plots of the SER as a function of the SNR at each receive antenna for the Silver code for the $4 \times 4$ system and the Perfect code. The Silver code for the $4 \times 4$ system nearly matches the Perfect code in performance at low and medium SNR. More importantly, it has lower ML-decoding complexity of the order of $M^{12.5}$, while that of the Perfect code is $M^{15}$.
\end{enumerate}
\subsection{8 Tx}
 To construct the Silver code for the $8 \times 2$ system, we first construct a rate-1, 4-group decodable STBC as described in Section \ref{sec4} and denote the set of obtained weight matrices by $\mathcal{G}_1$. Next we multiply the weight matrices of $\mathcal{G}_1$ by $\textbf{F}_4$ to obtain a new set of weight matrices which is denoted by $\mathcal{G}_2$. The weight matrices of the Silver code for the $8 \times 2$ system are obtained from $\mathcal{G}_1 \bigcup \mathcal{G}_2$. The Silver code for the $8 \times 3$ system  can be obtained by multiplying the matrices of $\mathcal{G}_1$ with $\textbf{F}_6$ and appending the resulting weight matrices to the set $\mathcal{G}_1 \bigcup \mathcal{G}_2$. The rival code is the punctured perfect code for 8 transmit antennas \cite{new_per}. The maximum mutual information plots of the two codes are shown in Fig. \ref{fig_cap1}. As expected, our code has higher mutual information, although lower than the ergodic capacity of the corresponding MIMO channels.

Fig. \ref{fig82} shows the symbol error performance of the Silver code for $8 \times 2$ system and the punctured Perfect code \cite{new_per}. The constellation employed is 4-QAM. Again, to enhance performance by way of increasing the coding gain, we have multiplied the weight matrices of $\mathcal{G}_2$ with the scalar $e^{\frac{j\pi}{4}}$, as done for the codes for 4 transmit antennas. The simulation plot suggests that our code has full diversity. The most important aspect of our code is that it has an ML-decoding complexity of $M^{9.5}$, while that of the comparable punctured Perfect code is $M^{15}$.

\section{Discussion} \label{sec7}

In this paper, we proposed a scheme to obtain full-rate STBCs for $2^a$ transmit antennas and any number of receive antennas with the lowest ML-decoding complexity and the least self-interference among known codes. The STBCs thus obtained allow higher mutual information than existing STBCs for the case $n_r < n_t$. Identifying explicit constellations which can be mathematically proven to guarantee full-diversity and a non-vanishing determinant without increasing the ML-decoding complexity is an open problem. Also, one can seek to obtain full-rate STBCs with reduced ML-decoding complexity for arbitrary number of transmit (not a power of 2) and receive antennas. These are some of the directions for future research. 

\section*{ACKNOWLEDGEMENT}
This work was partly supported by the DRDO-IISc program on Advanced Research in Mathematical Engineering through research grants and the INAE Chair Professorship to B. Sundar Rajan. We thank the anonymous reviewers for their useful comments which have greatly helped in enhancing the quality of the paper.

\newpage
\begin{figure}
\centering
\includegraphics[width=5.5in,height=3.5in]{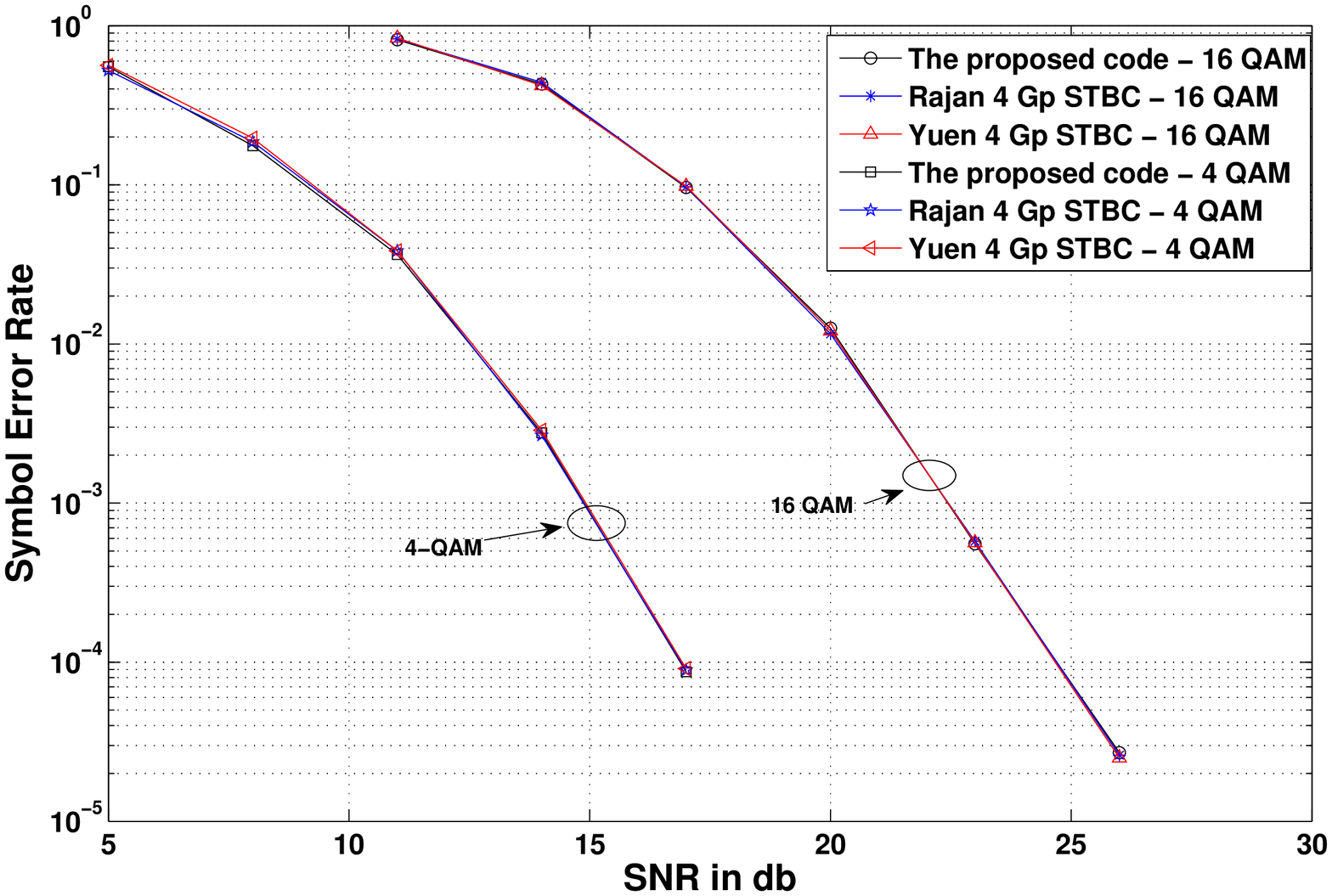}
\caption{SER comparison of the proposed STBC with a few known 4-group decodable STBCs for the $8\times1$ MIMO system}
\label{fig_4gp}
\end{figure}

\begin{figure}
\centering
\includegraphics[width=5.5in,height=3.5in]{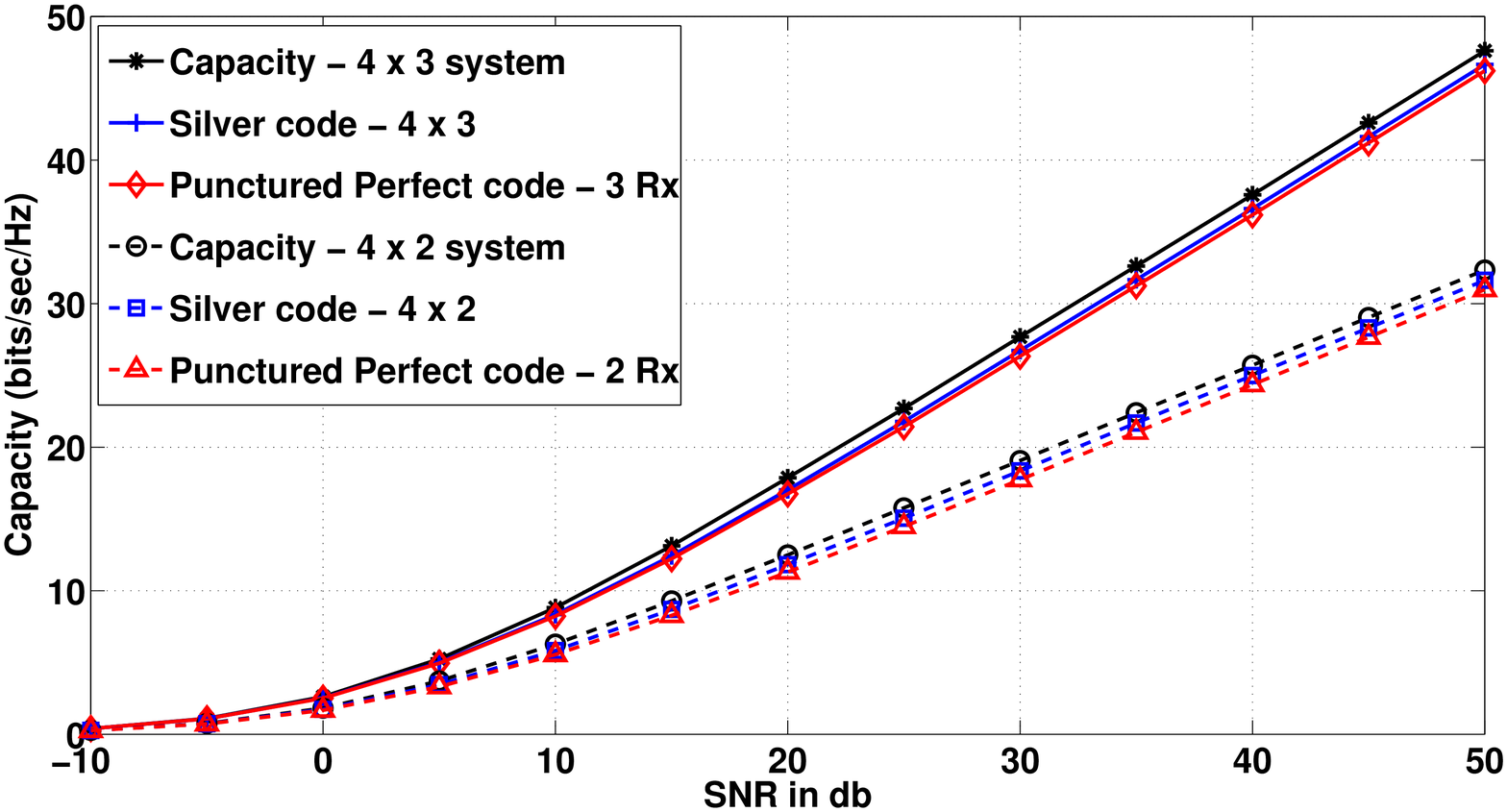}
\caption{Ergodic capacity Vs SNR for codes for $4 \times 2$ and $4 \times 3$ systems}
\label{fig_cap}
\end{figure}

\begin{figure}
\centering
\includegraphics[width=5.5in,height=3.5in]{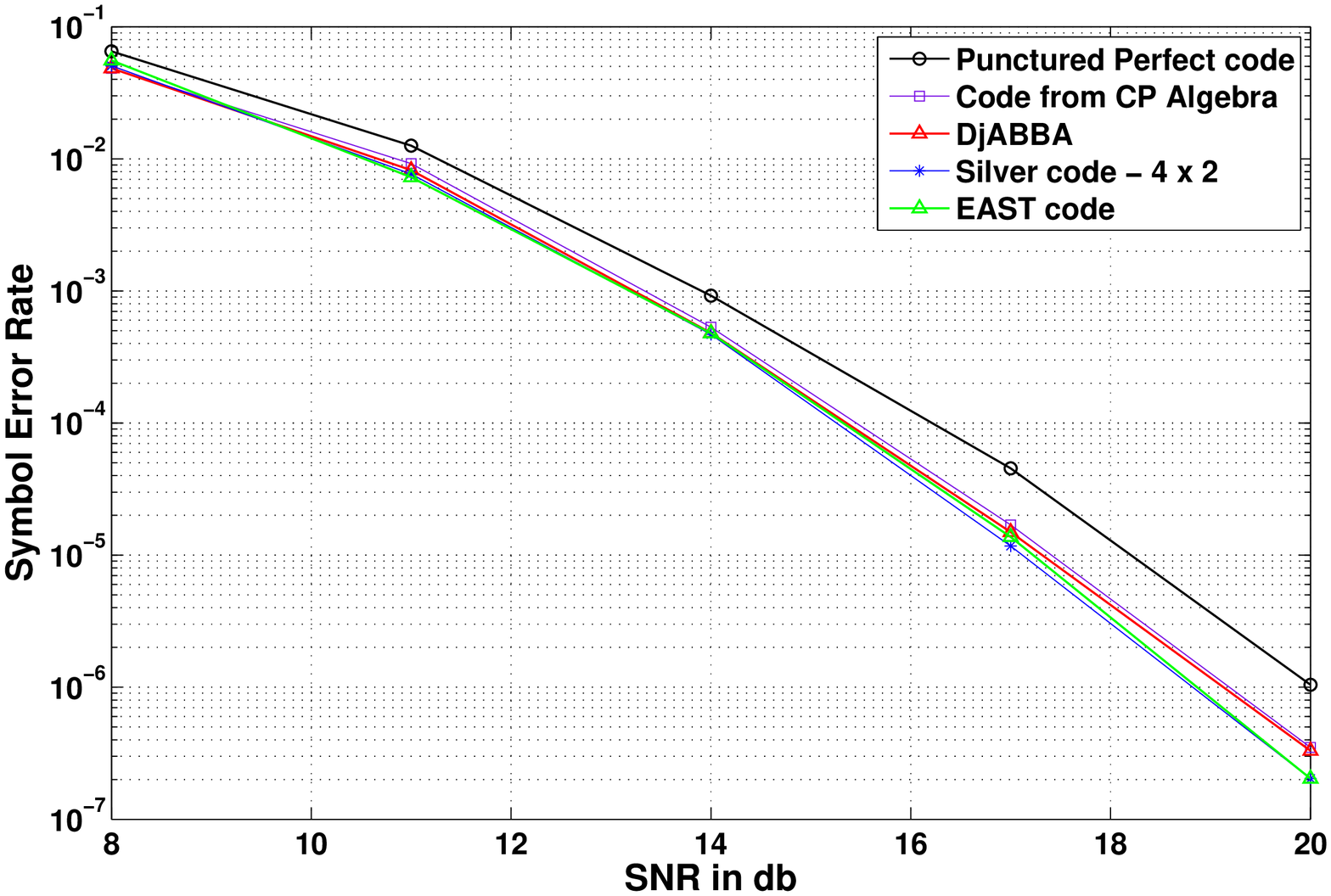}
\caption{SER performance at 4 BPCU for codes for $4 \times 2$ systems}
\label{fig1}
\end{figure}

\begin{figure}
\centering
\includegraphics[width=5.5in,height=3.5in]{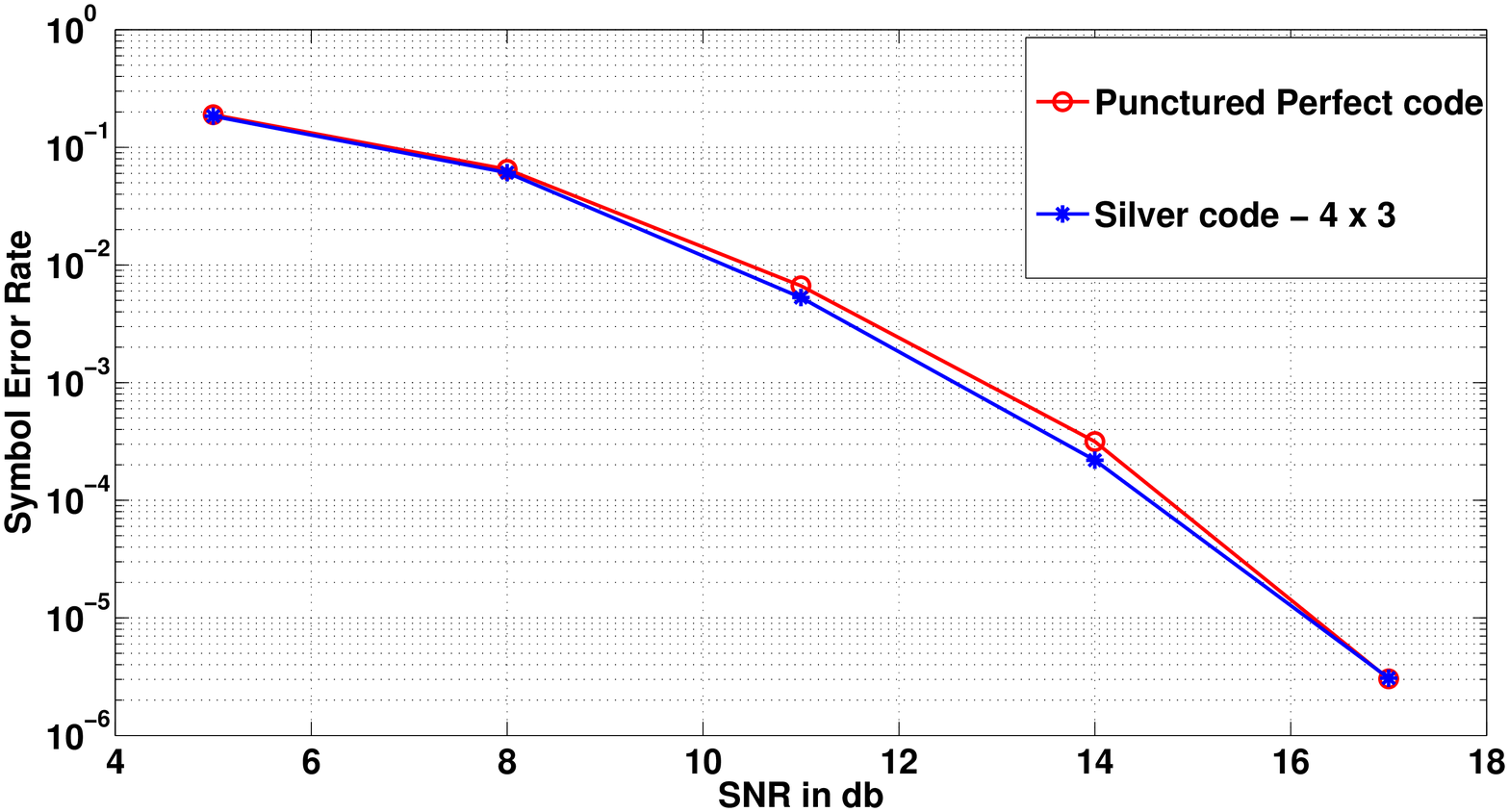}
\caption{SER performance at 6 BPCU for codes for $4 \times 3$ systems}
\label{fig2}
\end{figure}

\begin{figure}
\centering
\includegraphics[width=5.5in,height=3.5in]{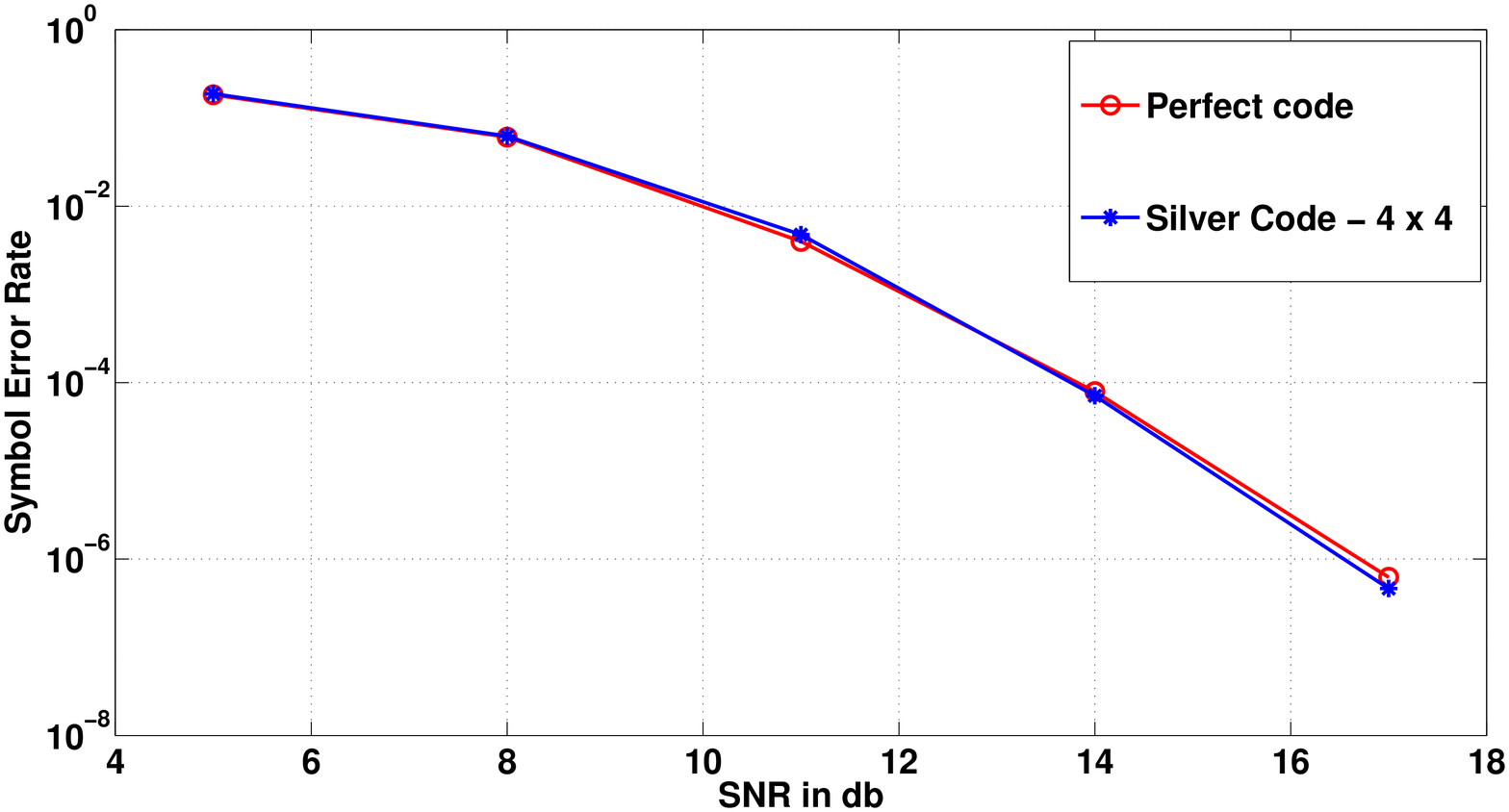}
\caption{SER performance at 8 BPCU for codes for $4 \times 4$ systems}
\label{fig3}
\end{figure}

\begin{figure}
\centering
\includegraphics[width=5.5in,height=3.5in]{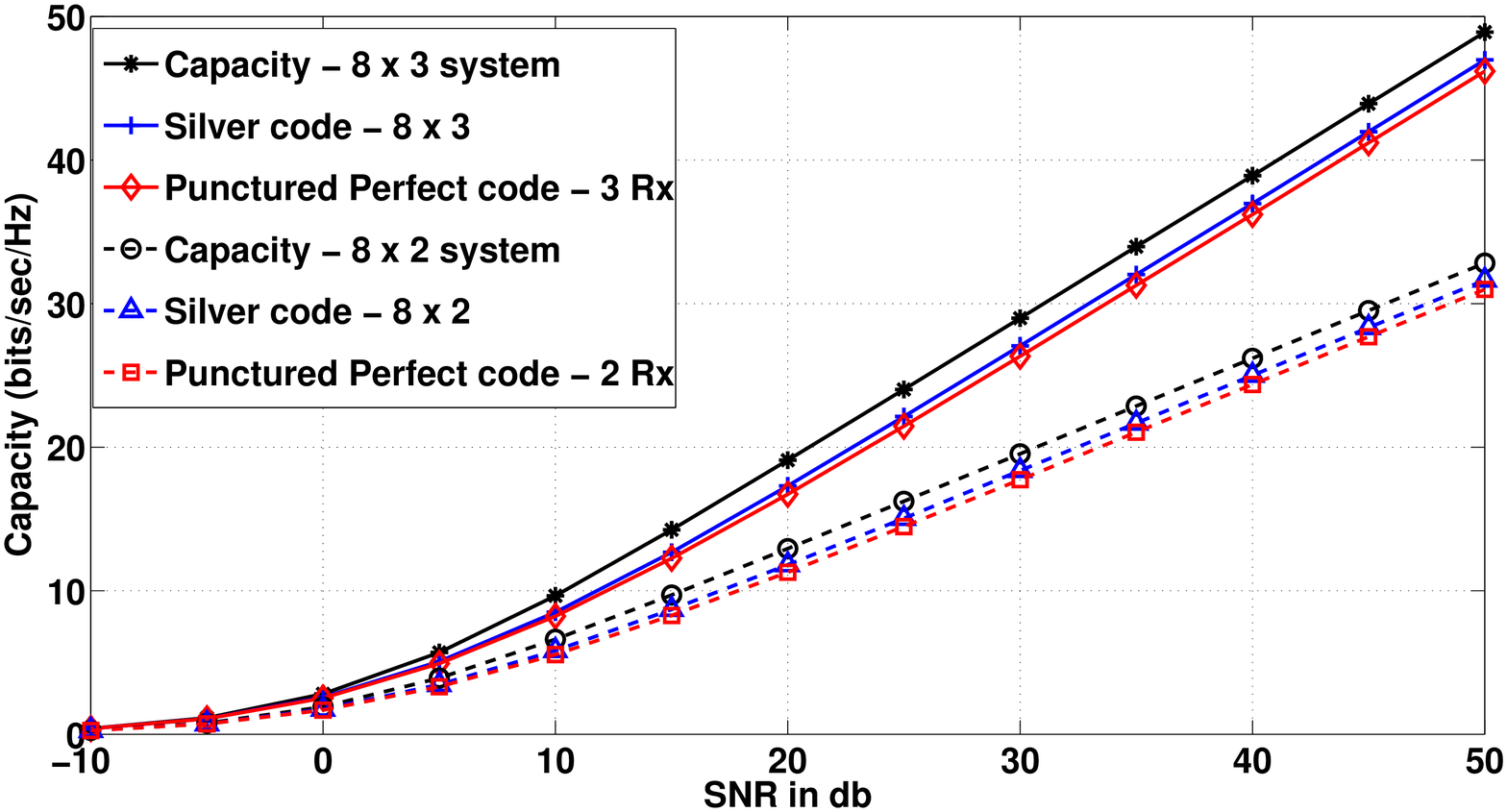}
\caption{Ergodic capacity Vs SNR for codes for $8 \times 2$ and $8 \times 3$ systems}
\label{fig_cap1}
\end{figure}

\begin{figure}
\centering
\includegraphics[width=5.5in,height=3.5in]{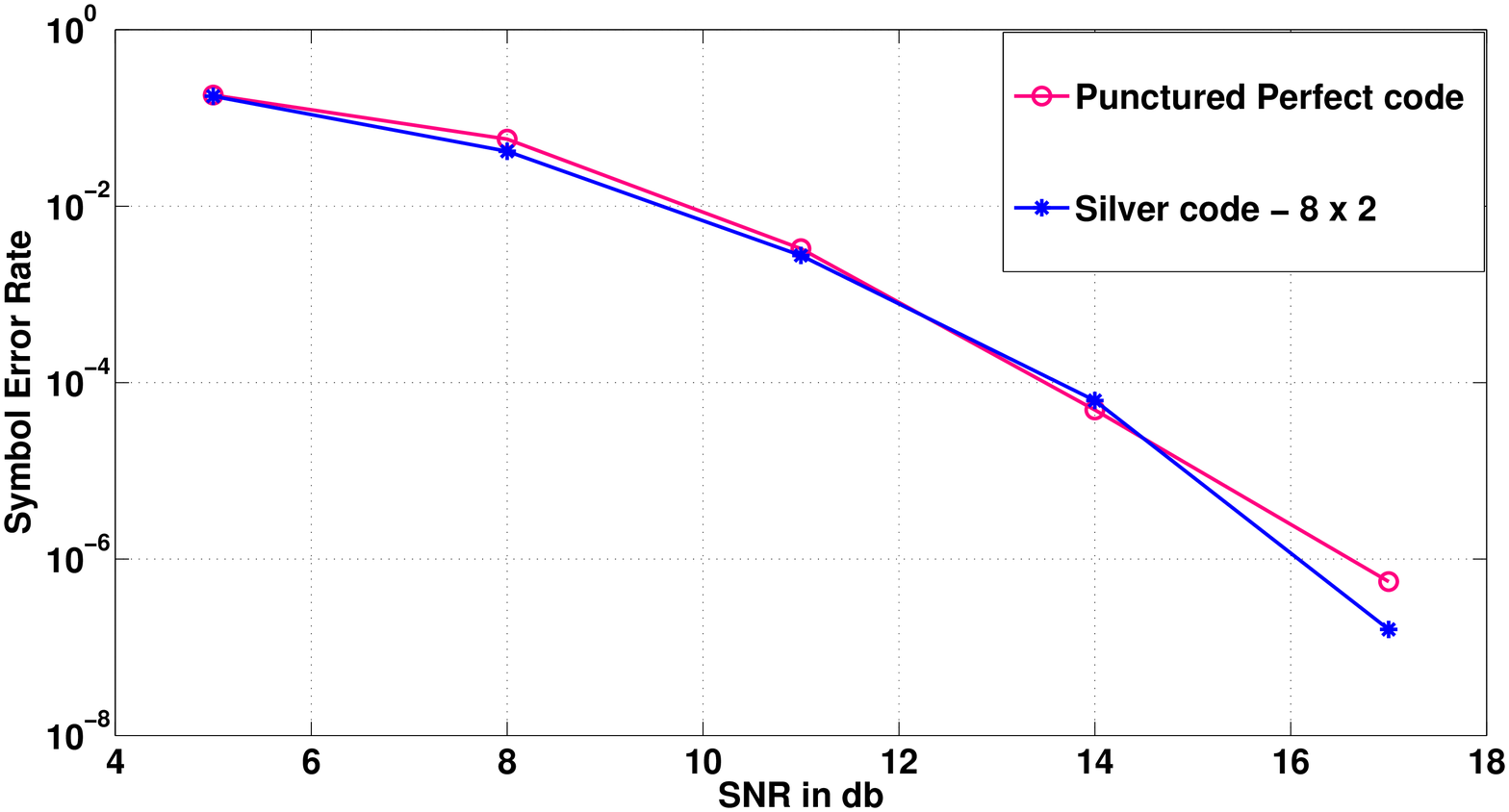}
\caption{SER performance at 4 BPCU for codes for $8 \times 2$ systems}
\label{fig82}
\end{figure}


\begin{thebibliography}{1}
\bibitem{TJC} V. Tarokh, H. Jafarkhani and A. R. Calderbank, ``Space-Time block codes from orthogonal designs,'' \emph{IEEE Trans. Inf. Theory,} vol.\ 45, no.\ 5, pp.\ 1456-1467, Jul. 1999. Also ``Correction to ``Space-time block codes from orthogonal designs,'' \emph{IEEE Trans. Inf. Theory}, vol.\ 46, no.\ 1, pp.\ 314, Jan. 2000.

\bibitem{TiH}
O. Tirkkonen and A. Hottinen, ``Square-matrix embeddable space-time block codes for
complex signal constellations,'' {\it IEEE Trans. Inf. Theory}, vol.\ 48, no.\ 2, pp.\ 384-395, Feb. 2002.

\bibitem{BRV}
J. C. Belfiore, G. Rekaya and E. Viterbo, ``The Golden Code: A $2\times2$ full rate space-time code with non-vanishing determinants,'' \emph{IEEE Trans. Inf. Theory}, vol.\ 51, no.\ 4, pp.\ 1432-1436, Apr. 2005.

\bibitem{SrR_arxiv}
K. P. Srinath and B. S. Rajan, ``Low ML-Decoding Complexity, Large Coding Gain, Full-Rate, Full-Diversity STBCs for $2\times 2$ and $4 \times 2$ MIMO Systems,'' \emph{IEEE Journal Sel. Topics Signal Process.}, vol.\ 3, no.\ 6, pp.\ 916-927, Dec. 2009.

\bibitem{john_barry1}
M. O. Sinnokrot and J. Barry, ``Fast Maximum-Likelihood Decoding of the Golden Code,'' \emph{IEEE Trans. Wireless Commun.}, vol.\ 9, no.\ 1, pp.\ 26-31, Jan. 2010.

\bibitem{Hollanti_silver}
C. Hollanti, J. Lahtonen, K. Ranto, R. Vehkalahti and E. Viterbo, ``On the algebraic structure of the Silver code: A $2\times2$ Perfect space-time code with non-vanishing determinant,'' in Proc. of \emph{IEEE Inf. Theory Workshop}, Porto, Portugal, May 2008.

\bibitem{tirk_combined}
O. Tirkkonen and R. Kashaev, ``Combined Information and Performance Optimization of Linear MIMO Modulations,'' in Proc. of \emph{ISIT 2002}, Lausanne, Switzerland, Jun. 30 - Jul. 8, 2002.

\bibitem{PGA}
J. Paredes,  A.B. Gershman,  M. G.-Alkhansari, `` 
A New Full-Rate Full-Diversity Space-Time Block Code With Nonvanishing Determinants and Simplified Maximum-Likelihood Decoding,'' \emph{IEEE Trans. Signal Process.,} vol.\ 56, no. 6, pp.\ 2461 - 2469, Jun. 2008.

\bibitem{Ray_silver}
A. Ray, K. Vinodh, G. R.-B. Othman and P. V. Kumar, ``Ideal structure of the silver code,'' in Proc. of \emph{ISIT 2009}, Seoul, South Korea, Jun. 28 - Jul. 03, 2009.

\bibitem{fixed_point_silver}
Y. Wu and L. Davis, ``Fixed-point fast decoding of the silver code,'' in \emph{IEEE Australian Communications Theory Workshop (AusCTW)}, 2009.

\bibitem{BHV}
E. Biglieri, Y. Hong and E. Viterbo, ``On Fast-Decodable Space-Time Block Codes,'' \emph{IEEE Trans. Inf. Theory}, vol.\ 55, no.\ 2, pp.\ 524-530, Feb. 2009.

\bibitem{HTW}
A. Hottinen, O. Tirkkonen and R. Wichman, ``Multi-antenna Transceiver Techniques for 3G and Beyond,'' Wiley publisher, UK, 2003.

\bibitem{JH}
H. Jafarkhani, ``A quasi-orthogonal space-time block code,'' \emph{IEEE Trans. Commun.}, vol.\ 49, no.\ 1, pp.\ 1-4, Jan. 2001.

\bibitem{ZS}
Z. A. Khan and B. S. Rajan, ``Single Symbol Maximum Likelihood Decodable Linear STBCs,'' \emph{IEEE Trans. Inf. Theory}, vol.\ 52, no.\ 5, pp.\ 2062-2091, May 2006.

\bibitem{Robert}
S. Sirianunpiboon, Y. Wu, A. R. Calderbank and S. D. Howard, ``Fast Optimal Decoding of Multiplexed Orthogonal Designs by Conditional Optimization,'' \emph{IEEE Trans. Inf. Theory}, vol.\ 56, no.\ 3, pp.\ 1106-1113, Mar. 2010.

\bibitem{fgd1}
 T. P. Ren, Y. L. Guan, and C. Yuen, ``Fast-group-decodable space-time block code,'' in Proc. of \emph{ITW 2010}, Cairo, Egypt, 2010.

\bibitem{fgd2}
 T. Ren, Y. Guan, C. Yuen, and E. Zhang, ``Space-time codes with block-orthogonal structure and their simplified ML and near-ML decoding,,'' in Proc. of \emph{VTC 2010-Fall}, Sep. 2010.

\bibitem{fgd3}
T. P. Ren, Y. L. Guan, C. Yuen, and E. Y. Zhang, ``Block-orthogonal space-time codes with decoding complexity reduction,'' in Proc. of \emph{SPAWC}, 2010.

\bibitem{Oggier_fast}
F. Oggier, R. Vehkalahti and C. Hollanti, ``Fast-decodable MIDO Codes
from Crossed Product Algebras,'' in Proc. of \emph{IEEE ISIT 2010}, Austin, Texas, Jun. 2010.

\bibitem{Oggier_spcom}
F. Oggier, C. Hollanti and R. Vehkalahti, ``An algebraic MIDO-MISO code construction,'' in Proc. of the \emph{Intl. conf. on signal
process. and commun. (SPCOM 2010)}, Bangalore, India, Jul. 2010.

\bibitem{damen_info}
M. O. Damen, A. Tewfik, and J. C. Belfiore, ``A construction of a space time
code based on number theory,'' \emph{IEEE Trans. Inf. Theory}, vol.\ 48,
no.\ 3, pp.\ 753-760, Mar. 2002.

\bibitem{tirk_expansion}
R. Kashaev and O. Tirkkonen, ``On expansion of MIMO mutual information in SNR,'' in Proc. of \emph{ISIT 2002}, Lausanne, Switzerland, Jun. 30 - Jul. 8, 2002.

\bibitem{hollanti1}
C. Hollanti and H. F. Lu, ``Construction methods for asymmetric and multi-block space-time codes,'' \emph{IEEE Trans. Inf. Theory}, vol. 55, no. 3, pp. 1086 – 1103, Mar. 2009.

\bibitem{hollanti2}
 H. F. Lu and C. Hollanti, ``Optimal diversity multiplexing tradeoff and code constructions of constrained asymmetric MIMO systems,'' \emph{IEEE Trans. Inf. Theory}, vol. 56, no. 5, pp. 2121-2129, May 2010.

\bibitem{ORBV}
F. Oggier, G. Rekaya, J. C. Belfiore and E. Viterbo, ``Perfect space time block codes,'' \emph{IEEE Trans. Inf. Theory,} vol.\ 52, no.\ 9, pp.\ 3885-3902, Sep. 2006.

\bibitem{new_per}
P. Elia, B. A. Sethuraman and P. V. Kumar, ``Perfect Space-Time Codes for Any Number of Antennas,'' \emph{IEEE Trans. Inf. Theory}, vol. 53 , no.\ 11, pp. 3853-3868, Nov. 2007.

\bibitem{4gp1}
D. N. Dao, C. Yuen, C. Tellambura, Y. L. Guan, and T. T. Tjhung, ``Four-group decodable space-time block codes," \emph{IEEE Trans. Signal Process.}, vol.\ 56, no.\ 1, pp.\ 424-430, Jan. 2008.

\bibitem{4gp2}
G. S. Rajan and B. S. Rajan, ``Multi-group ML Decodable Collocated and Distributed Space Time Block Codes,''  \emph{IEEE Trans. Inf. Theory}, vol.\ 56, no.\ 7, pp.\ 3221-3247, Jul. 2010.

\bibitem{sanjay}
S. Karmakar and B. S. Rajan, ``Multigroup-Decodable STBCs from Clifford Algebras,''
\emph{IEEE Trans. Inf. Theory}, vol.\ 55, no.\ 1, pp.\ 223-231, Jan.\ 2009.

\bibitem{HaH}
B. Hassibi and B. Hochwald, ``High-rate codes that are
linear in space and time,'' {\it IEEE Trans. Inf. Theory}, vol.\ 48, no.\ 7, pp.\ 1804-1824, July 2002.

\bibitem{tel}
I. E. Telatar, ``Capacity of multi-antenna Gaussian channels,'' \emph{Eur. Trans. Telecom.}, vol.\ 10, pp. \ 585
595, Nov. 1999.

\bibitem{JJK}
J. K. Zhang, J. Liu, K. M. Wong, ``Trace-Orthonormal Full-Diversity Cyclotomic Space Time Codes,'' \emph{IEEE Trans. Signal Process.}, vol.\ 55, no.\ 2, pp.\ 618-630, Feb. 2007.

\bibitem{anti_matric}
D. B. Shapiro and R. Martin, ``Anticommuting Matrices,'' \emph{The American Mathematical Monthly}, vol.\ 105, no.\ 6, pp.\ 565-566, Jun. -Jul., 1998.


\bibitem{full_div_rot}
http://www1.tlc.polito.it/~viterbo/rotations/rotations.html.

\bibitem{lakshmi}
L. P. Natarajan and B. S. Rajan,
``Asymptotically-Optimal, Fast-Decodable, Full-Diversity STBCs,'' available online at arXiv, arXiv:1003.2606v2, Aug. 20, 2010.

\bibitem{sphere_decoding}
E. Viterbo and J. Boutros, ``Universal lattice code decoder for fading channels,'' \emph{IEEE Trans. Inform theory.}, vol.\ 45, no.\ 5, pp. \ 1639-1642, Jul. 1999.

\bibitem{barry}
M. O. Sinnokrot, J. R. Barry and V. K. Madisetti, ``Embedded Alamouti Space-Time Codes for High Rate and Low Decoding Complexity,'' \emph{IEEE Asilomar 2008}.

\end{thebibliography}
\end{document}